\documentclass[aoas,preprint]{imsart}
\RequirePackage[OT1]{fontenc}
\RequirePackage{amsthm,amsmath,mathrsfs,multirow,morefloats,booktabs,subfigure,pdfpages,color,array,hhline,makecell,epstopdf,algorithm}
\RequirePackage[colorlinks,citecolor=blue,urlcolor=blue]{hyperref}
\RequirePackage{algpseudocode}
\algrenewcommand\algorithmicrequire{\textbf{Input:}}
\algrenewcommand\algorithmicensure{\textbf{Output:}}
\usepackage{amssymb}

\startlocaldefs
\numberwithin{equation}{section}
\theoremstyle{plain}
\newtheorem{thm}{Theorem}[section]
\newtheorem{defin}[thm]{Definition}
\newtheorem{lem}[thm]{Lemma}

\newtheorem{rem}[thm]{Remark}
\endlocaldefs
\allowdisplaybreaks[4]

\usepackage{titlesec}
\titleformat*{\section}{\bfseries\Large\rmfamily}
\titleformat*{\subsection}{\bfseries\large\rmfamily}
\titleformat*{\subsubsection}{\bfseries\normalsize\rmfamily}
\begin{document}
\begin{frontmatter}
\title{\Large
An improved spectral clustering method
for mixed membership community detection}
\runtitle{Mixed-ISC for mixed membership community detection}
%\thankstext{T1}{Footnote to the title with the ``thankstext'' command.}

\begin{aug}
\author{\fnms{Huan} \snm{Qing}
\thanksref{ m1}
\ead[label=e1]{qinghuan@cumt.edu.cn}}
\and
\author{\fnms{Jingli} \snm{Wang}
\thanksref{ m2}
\ead[label=e2]{jlwang@nankai.edu.cn}}

%\thankstext{t1}{Support}
%\thankstext{t2}{Support}
\runauthor{H.Qing and J.Wang}

\affiliation{China University of Mining and Technology\thanksmark{m1}
\and
Nankai University\thanksmark{m2}
}
\today
\address[A]{
H.Qing\\
Department of Mathematics\\
China University of Mining and Technology\\
Xuzhou, 221116\\
 P.R. China\\
\printead{e1}\\
\phantom{E-mail:\ }}

\address[B]{
J.Wang\\
School of Statistics and Data Science\\
Nankai University\\
Tianjin, 300071\\
 P.R. China\\
\printead{e2}\\
\phantom{E-mail:\ }}
\end{aug}
\begin{abstract}
%Abstract is about 200 characters long.

%Paragraph 1. Data and problem met

%Paragraph 2. Our contribution\\
\begin{center}
\textbf{Abstract}
\end{center}
 Community detection has been well studied recent years, but the more realistic case of
mixed membership community detection remains a challenge. Here, we develop an efficient spectral algorithm Mixed-ISC based on applying more than $K$ eigenvectors for clustering given $K$ communities for estimating the community memberships under the degree-corrected mixed membership (DCMM) model. We show that the algorithm is asymptotically consistent. Numerical experiments on both simulated networks and many empirical networks demonstrate that Mixed-ISC performs well compared to a number of benchmark methods for mixed membership community detection. Especially, Mixed-ISC provides satisfactory performances on weak signal networks.
\end{abstract}
%\begin{keyword}[class=MSC]
%\kwd[Primary ]{60K35}
%\kwd{60K35}
%\kwd[; secondary ]{60K35}
%\end{keyword}
\end{frontmatter}
%\tableofcontents
\section{Introduction}
%Introduce the background of the problem, including the data type and some previous results
In the study of networks, including physics, computer science, statistics, and the social science, the problem of (mixed membership) community detection has received substantial attentions, see \cite{anandkumar2014tensor,Cai2016survey,Fortunato2016Community, goldenberg2010a,SBM,SCORE, SLIM,Tutorial, GeoNMF,RSC, OCCAM}. Consider an undirected, unweighted, no-self-loop network $\mathcal{N}$, and let $A$ be its adjacency matrix such that $A_{ij}=1$ if there is an edge between node $i$ and $j$, $A_{ij}=0$ otherwise. Since there is no-self-loop in $\mathcal{N}$, all diagonal entries of $A$ are zero. In the perspective of spectral clustering, the number of communities $K$ can be determined empirically by the distinguished leading eigenvalues of the adjacency matrix $A$ or its variants.

In community detection area, there are two classical real networks, Simmons and Caltech \cite{traud2011comparing}. The Simmons network contains 1137 nodes in which nodes are students with graduation year between 2006 and 2009. It is deemed as four friendship communities such that students in the same year are in the same community, which is applied as the ground truth. The Caltech network contains a largest component with 590 nodes and the community structure is highly correlated with the eight dorms since students in the same dorm are more likely to be friends. Therefore, the Caltech network contains eight communities.
 However, a number of benchmark community detection methods detect the communities of this two networks with large error rates.  For example, the error rates of SCORE \cite{SCORE}, RSC \cite{RSC}, CMM \cite{CMM} and OCCAM \cite{OCCAM} are 268/1137, 244/1137, 137/1137, 266/1137 for Simmons, and 180/590, 170/590, 106/590, 189/590 for Caltech, respectively. Encouragingly, Jin et al. \cite{SCORE+} found that the $(K+1)$-th leading eigenvalues of the adjacency matrices for this two networks are close to their $K$-th eignvalues and they also contain label information. They called this kind of networks as ``weak signal networks". Therefore we give the following definition:
\begin{defin}\label{weakA}
	Given an un-weighted, undirected, no-self-loops and connected network $\mathcal{N}$ with adjacency matrix $A$ and the number of communities $K$. The network is a weak signal network if $1-|\frac{\lambda^{A}_{K+1}}{\lambda^{A}_{K}}|\leq 0.1$, otherwise it is strong signal network, where $\lambda^{A}_{k}$ denotes the $k$-th leading eigenvalue of $A, k=K, K+1$.
\end{defin}
It needs to note that we measure the ``signal weakness'' slightly different as that in \cite{SCORE+}, where we use $1-|\frac{\hat{\lambda}_{K+1}}{\hat{\lambda}_{K}}|$ instead of the $1-\frac{\hat{\lambda}_{K+1}}{\hat{\lambda}_{K}}$ in \cite{SCORE+} since we find that the leading eigenvalues are measured by magnitude, which means that $\hat{\lambda}_{K+1}$ may have different sign as that of $\hat{\lambda}_{K}$. When they have different signs but their absolute values are close, the $(K+1)$-th eigenvector still has label information.
By considering $(K+1)$ eignvalues, Jin et al \cite{SCORE+} extended the classical method SCORE \cite{SCORE} for weak signal networks, and their method is known as SCORE+. SCORE+ has pleasure detection results since the error rates are 127/1137 for Simmons and 98/590 for Caltech. Qing and Wang \cite{ISC} inspired by this idea and constructed the Improved Spectral Clustering (ISC) method with full theoretical guarantee  which also has very low error rates, with 121/1137 for Simmons and 96/590 for Caltech.

When we study the mixed membership networks, we find that there are a large number of networks which are weak signal networks, such as the Stanford Network Analysis Project (SNAP) ego-networks \cite{SNAPego,OCCAM} and Coauthorship networks \cite{ji2016coauthorship}.
The SNAP ego-networks dataset contains substantial ego-networks from Facebook, Twitter, and GooglePlus. In an ego-network, all the nodes are friends of one central user and the friendship groups or circles (depending on the platform) set by the users can be used as ground truth communities. Meanwhile, nodes in SNAP ego-networks may have overlapping memberships (i.e., the true overlapping memberships are known for SNAP ego-networks.). We obtain the SNAP ego-networks parsed by Yuan Zhang (the first author of \cite{OCCAM}). The parsed SNAP ego-networks are slightly different from those used in Zhang et al. (2020). We find that there exist many weak signal networks in SNAP ego-networks, see Table \ref{SNAPweak}. Therefore, it is necessary and important to design method that can deal with weak signal networks for mixed membership community detection.
\begin{table}[h!]
	\centering
	\caption{Numbers of weak signal and strong signal networks in SNAP ego-networks.}
	\label{SNAPweak}
	%\resizebox{\linewidth}{!}{
	\begin{tabular}{cccccccccc}
		\toprule
		&\#Networks&\# Weak Signal networks&\# Strong Signal networks\\
		\midrule
		Facebook&7&0&7\\
		\hline
		GooglePlus &58&13&45\\
		\hline
		Twitter&255&99&156\\
		\bottomrule
	\end{tabular}%}
\end{table}
%Since for weak signal networks, the $(K+1)$-th leading eigenvector of $A$ or its variants may also contain information about nodes memberships and there exist plenty of weak signal networks in SNAP ego-networks, which motivates us to design method that can deal with weak signal networks for mixed membership community detection.

Ji and Jin \cite{ji2016coauthorship} collected a coauthorship network data set for statisticians, based on all published papers in AOS, Biometrika, JASA, JRSS-B, from 2003 to the first half of 2012. A large number of researches have studied this dataset \cite{li2020hierarchical,karwa2016discussion}. In this network, an edge is constructed between two authors if they have coauthored at least two papers in the range of the data set. As suggested by \cite{mixedSCORE}, there are two communities called ``Carroll-Hall'' and ``North Carolina'' over 236 nodes (i.e., $n=236, K=2$ for Coauthorship network), and authors in this network have mixed memberships in the two communities. We find that the Coauthorship network is a weak signal network since $1-|\frac{\lambda^{A}_{3}}{\lambda^{A}_{2}}|=0.0787$ where $A$ is the adjacency matrix of the Coauthorship network. Based on this observation, we argue that method which can deal with weak signal networks for mixed membership community detection may provide some new insights on the analysis of the Coauthorship network.

%%%%%%%%%%%%%%%%%%%%%%%%%%%%%%%%%%%%%%%%%%%%%%%%%%%%%
 %Since mixed membership community detection problem is an extension of community detection problem, we argue that methods for mixed membership community detection should detect networks based on community detection. In \cite{SCORE+}, the authors find that two networks Simmons and Caltech are weak signal networks (see Definitions \ref{weakA} and \ref{weakL}) such that classical community detection methods, such as CMM \cite{CMM}, SCORE \cite{SCORE}, and RSC \cite{RSC}, may detect this two networks with high error rates. In this paper, we find that some methods for mixed membership community detection such as Mixed-SCORE \cite{mixedSCORE}, OCCAM \cite{OCCAM} and GeoNMF \cite{GeoNMF} also detect this two weak signal networks with high error rates. Meanwhile, Tables \ref{SNAPweak} and \ref{SNAPweakL} show that weak signal networks are  ubiquitous in empirical datasets. These findings motivates us to design a new approach which can successfully detect weak signal networks for mixed membership community detection.

To our knowledge, the study of such weak signal network with mixed membership had not been reported.
Inspired by the ISC method \cite{ISC}, which can successfully detect weak signal networks as well as strong signal networks, we extend the ISC method to mixed membership community detection, and name the new approach as Mixed-ISC. Mixed-ISC is different from ISC in important ways:
\begin{itemize}
  \item Models are different. ISC is for non-mixing models and it is designed based on the degree-corrected stochastic block model (DCSBM) \cite{DCSBM}, while Mixed-ISC is for mixing models and designed based on the degree-corrected mixed membership (DCMM) model \cite{mixedSCORE}.
  \item Problems are different. ISC is for the community detection problem while Mixed-ISC is a mixed membership community detection method.
  \item Methodologies are different. There are a cluster center hunting (CCH) step and a membership reconstruction (MR) step in Mixed-ISC while there is no such steps in ISC.
\end{itemize}

Overall, Mixed-ISC can apply more than $K$ eigenvectors and eigenvalues for reconstructing nodes memberships, while traditional mixed membership community detection methods such as Mixed-SCORE \cite{mixedSCORE}, OCCAM \cite{OCCAM} and GeoNMF \cite{GeoNMF}, can not. We show that our method is asymptotically consistent and it performs satisfactory compared to a number of benchmark methods for a small scale of simulations and substantial empirical datasets. %Meanwhile, our Mixed-ISC inherits nice property of ISC method such that Mixed-ISC is insensitive to the choice of tuning parameters.

This paper is organized as follows. In Section \ref{sec2}, we set up the mixed membership community detection problem under DCMM. In Section \ref{sec4}, we propose our Mixed-ISC method. Section \ref{sec5} presents theoretical framework for Mixed-ISC. Sections \ref{sec6} and \ref{sec7} investigates the performances of Mied-ISC methods via comparing with three benchmark methods on simulations and substantial empirical datasets. We also study the choice of tuning parameter in Section \ref{sec7}. Section \ref{sec8} concludes.
\section{Settings and Model}\label{sec2}

The following notations will be used throughout the paper: $\|\cdot\|_{F}$ for a matrix denotes the Frobenius norm,  $\|\cdot\|$ for a matrix denotes the spectral norm,  and $\|\cdot\|$ for a vector denotes the $l_{2}$-norm. For convenience, when we say ``leading eigenvalues'' or ``leading eigenvectors'', we are comparing the \emph{magnitudes} of the eigenvalues and their respective eigenvectors with unit-norm. For any matrix or vector $x$, $x'$ denotes the transpose of $x$. For any matrix $X$, set the matrix $\mathrm{max}(X,0)$ such that its $(i,j)$ entry is $\mathrm{max}(X_{ij},0)$.

Consider an undirected, unweighted, no-self-loop network $\mathcal{N}$ and assume that there are $K$ disjoint blocks $\mathcal{C}^{(1)}, \mathcal{C}^{(2)}, \ldots, \mathcal{C}^{(K)}$ where $K$ is assumed to be known. %Let $A$ be its adjacency matrix such that $A_{ij}=1$ if there is an edge between node $i$ and $j$, $A_{ij}=0$ otherwise. Since there is no-self-loop in $\mathcal{N}$, all diagonal entries of $A$ are zero.
In this paper, we use the degree-corrected mixed membership (DCMM) model \cite{mixedSCORE} to construct the mixed membership network. DCMM assumes that for each node $i$, there is a Probability Mass Function (PMF) $\pi_{i}=(\pi_{i}(1), \pi_{i}(2), \ldots, \pi_{i}(K))$ such that
\begin{align*}
	\mathrm{Pr}(i\in \mathcal{C}^{(k)})=\pi_{i}(k), \qquad 1\leq k\leq K, 1\leq i\leq n,
\end{align*}
which means that node $i$ belongs to cluster $\mathcal{C}^{(k)}$ with probability $\pi_{i}(k)$ \footnote{Here, $\pi_{i}(k)$ can also be understood as the weight of node $i$ belongs to community $k$.}. Hence we have $\sum_{k=1}^{K}\pi_{i}(k)=1$. In this sense, the mixed membership model allows us to measure the probability that each node belongs to a certain community. If a node only belongs to one community, we call this node $i$ as ``pure", and in this case $\pi$ is degenerate such that only one element of $\pi$ is 1, and all others $K-1$ entries are 0; and if a node shares among communities, it is called ``mixed". Furthermore, set $\underset{1\leq k\leq K}{\mathrm{max}}\pi_{i}(k)$ as the \textit{purity} of node $i$, for $1\leq i\leq n$.

For mixed membership community detection, the aim is to  estimate $\pi_{i}$ for all nodes $i\in\{1,2,\ldots, n\}$. First of all, we should model the adjacency matrix. To do so, we first define a  $K\times K$ symmetric non-negative matrix $P$ (called \textit{mixing matrix} in this paper) which is non-singular, irreducible, and $P(i,j)\in [0,1] \mathrm{~for~}1\leq i,j\leq K$. We also model the \textit{degree heterogeneity} by a positive vector $\theta=(\theta(1), \ldots, \theta(n))'$. For any fixed pair of  $(i,j)$, DCMM has the following assumption:
\begin{align*}
\mathrm{Pr}(A(i,j)=1|i\in V^{(k)}, j\in V^{(l)})=\theta(i)\theta(j)P(k,l),
\end{align*}
which means that under the condition that we know the information $i\in V^{(k)}, j\in V^{(l)}$, the probability that there is an edge between nodes $i,j$ is $\theta(i)\theta(j)P(k,l)$. For $1\leq i<j\leq n$, $A(i,j)$ are Bernoulli random variables that are independent of each other, satisfying
\begin{align*}
\mathrm{Pr}(A(i,j)=1)=\theta(i)\theta(j)\sum_{k=1}^{K}\sum_{l=1}^{K}\pi_{i}(k)\pi_{j}(l)P(k,l).
\end{align*}

 For convenience, introduce an $n\times K$ \textit{membership matrix} $\Pi$ such that the $i$-th row of $\Pi$ (denoted as $\Pi_{i}$) is $\pi_{i}$ for all $i\in \{1,2,\ldots, n\}$, and denote $\Theta$ as an $n\times n$ matrix whose $i$-th diagonal entry is $\theta_{i}$. As presented in \cite{mixedSCORE}, let $E[A]=\Omega$ such that $\Omega(i,j)=\mathrm{Pr}(A(i,j)=1), 1\leq i<j\leq n$. Then we have
\begin{align*}
	\Omega=\Theta \Pi P \Pi' \Theta.
\end{align*}

Hereafter, given $(n, P, \Theta, \Pi)$, we can generate the random adjacency matrix $A$ under DCMM, hence we denote the DCMM model as $DCMM(n, P, \Theta, \Pi)$ for convenience in this paper. In the case that when all nodes are pure, then DCMM reduces to the degree-corrected stochastic block model (DCSBM) \cite{DCSBM}. In the case where $\theta(i)=\theta(j)=c_{0}$ (a positive constant) for all modes $i,j\in\{1,2,\ldots, n\}$, then DCMM degenerates as MMSB \cite{MMSB}. Moreover, the identifiability of the DCMM model is studied in \cite{jin2017a, mixedSCORE}, hence the model is well defined.

In this paper, the primary goal is to estimate the membership matrix $\Pi$ with given $(A, K)$ assuming that $A$ is generated from the DCMM model. Finally, $\Theta$ is nuisance that we want to remove or reduce  the effects of it in a mixed membership community detection.

\section{Methodology}\label{sec4}
%\subsection{The algorithm: Mixed-ISC}
In this section we present the following algorithm which we call Mixed Improved Spectral Clustering (Mixed-ISC for short) method. %, is an extension of the ISC method \cite{ISC} to the mixed membership community detection problem.

%\noindent\rule{13cm}{0.4pt}

%\textbf{Mixed-ISC}. Input:  $A, K$, and a ridge regularizer $\tau\geq 0$. Output: $\hat{\Pi}$.
%\noindent\rule{13cm}{0.4pt}
With known $A, K$, and a ridge regularizer $\tau\geq 0$, we first use a normalization procedure, i.e., ISC step, to reduce the noise.

\vspace{0.4cm}
$\bullet$ \texttt{ISC step}:

1. Obtain the graph Laplacian with ridge regularization by
\begin{align*}
  L_{\tau}=D_{\tau}^{-1/2}AD_{\tau}^{-1/2},
\end{align*}
where $D_{\tau}=D+\tau I$, $D$ is an $n\times n$ diagonal matrix whose  $i$-th diagonal entry is $D(i,i)=\sum_{j=1}^{n}A(i,j)$,  $\tau= cd, d=\frac{d_{\mathrm{max}}+d_{\mathrm{min}}}{2}, d_{\mathrm{max}}=\mathrm{max}_{1\leq i\leq n}d_{i}, d_{\mathrm{min}}=\mathrm{min}_{1\leq i\leq n}d_{i}$ (where $d_{i}$ is the degree of node $i$, i.e., $d_{i}=\sum_{j=1}^{n}A_{ij}$ for $1\leq i\leq n$), $c$ is a nonnegative constant, and $I$ is the $n\times n$ identity matrix. Conventional choice for $c$ is 0.1.

2. Compute the leading $\textbf{K+1}$ eigenvalues and eigenvectors with unit norm of $L_{\tau}$, and then calculate  the weighted eigenvectors matrix:
\begin{align*}
  \hat{X}=[\hat{\eta}_{1},\hat{\eta}_{2}, \ldots, \hat{\eta}_{K}, \hat{\eta}_{K+1}]\cdot \mathrm{diag}(\hat{\lambda}_{1}, \hat{\lambda}_{2}, \ldots, \hat{\lambda}_{K}, \hat{\lambda}_{K+1}),
\end{align*}
where $\hat{\lambda}_{i}$ is the $i$-th leading eigenvalue of $L_{\tau}$, and $\hat{\eta}_{i}$ is the respective eigenvector with unit-norm, for $i = 1, 2, \cdots, K+1$.

3. Normalizing each row of $\hat{X}$  to have unit length, and denote by $\hat{X}^{*}$, i.e.,
$$\hat{X}^{*}_{ij}=\hat{X}_{ij}/(\sum_{j=1}^{K+1}\hat{X}_{ij}^{2})^{1/2}, i = 1, \dots, n, j=1, \dots, K+1.$$

After the ISC step, we obtain a normalized matrix $\hat{X}^*$, then we estimate communities centers by K-median algorithm, finally, we  reconstruct the membership with the estimated $K$ centers.

$\bullet$ \texttt{Cluster Centers Hunting (CCH) step}:

4. Perform K-median algorithm on the rows of $\hat{X}^{*}$ and obtain $K$ estimated cluster centers $\hat{v}_{1}, \hat{v}_{2}, \ldots, \hat{v}_{K} \in\mathcal{R}^{1\times (K+1)}$, i.e.,
\begin{align}
\{\hat{v}_{1}, \hat{v}_{2}, \ldots, \hat{v}_{K}\}=\mathrm{arg~}\underset{\hat{u}_{1},  \ldots, \hat{u}_{K}}{\mathrm{min}}\frac{1}{n}\sum_{i=1}^{n}\underset{\hat{u}\in\{\hat{u}_{1},\ldots, \hat{u}_{K}\}}{\mathrm{min}}\|\hat{X}^{*}_{i}-\hat{u}\|_{2}.
\end{align}
Form the $K\times (K+1)$ matrix $\hat{V}$ such that the $i$-th row of $\hat{V}$ is $\hat{v}_{i}, 1\leq i\leq K$.

$\bullet$ \texttt{Membership Reconstruction (MR) step}:

5. Project the rows of $\hat{X}^{*}$ onto the spans of $\hat{v}_{1}, \ldots, \hat{v}_{K}$, i.e., compute the $n\times K$ matrix $\hat{Y}$ such that $\hat{Y}=\hat{X}^{*}\hat{V}'(\hat{V}\hat{V}')^{-1}$. Set $\hat{Y}=\mathrm{max}(0, \hat{Y})$ and estimate $\pi_{i}$ by $\hat{\pi}_{i}=\hat{Y}_{i}/\|\hat{Y}_{i}\|_{1}, 1\leq i\leq n$.  Obtain the estimated membership matrix $\hat{\Pi}$ such that its $i$-th row is $\hat{\pi}_{i}, 1\leq i\leq n$.
\vspace{0.4cm}
%\noindent\rule{13cm}{0.4pt}

Several remarks about Mixed-ISC method are listed in order.
\begin{itemize}
  %\item The ISC step aims at computing the $n\times (K+1)$ matrix $\hat{X}^{*}$ in the ISC method \cite{ISC}, that's why we state that our Mixed-ISC method is designed based on the ISC approach.
  \item Similar as \cite{OCCAM}, we also apply K-median clustering method to hunt the $K$ centers of  $\hat{X}^{*}$. As is known, K-means method is usually applied in spectral clustering community detection algorithms \cite{javed2018community}. However, in mixed membership community detection problems, K-means may fail to find the mixture cluster centers, while K-median can ignore some marginal nodes and have a chance to detect correct centers. With K-median spectral clustering, we can still get consistent estimation \cite{lei2015consistency,OCCAM}. However, by numerical studies, we found that K-median clustering is much slower than K-means, but they have same results in almost all scenarios. Therefore, for time-saving, we still use K-means clustering method in simulation studies and real data analysis in section \ref{sec6} and \ref{sec7}.
  %Unlike applying K-medians in OCCAM \cite{OCCAM} or vertex hunting algorithm in Mixed-SCORE \cite{mixedSCORE} for hunting the $K$ centers of $\hat{X}^{*}$, we find that its enough for our Mixed-ISC to apply K-means in the CCH step, and it performs satisfactory both numerically and empirically. Meanwhile, for the convenience of theoretical analysis for Mixed-ISC, we apply K-medians for its theoretical part following similar process as in the OCCAM method \cite{OCCAM}.
  \item In the MR step, we compute $\hat{Y}$ by setting it as $\hat{X}^{*}\hat{V}'(\hat{V}\hat{V}')^{-1}$, which is quite different from that of Mixd-SCORE \cite{mixedSCORE} and OCCAM \cite{OCCAM}, and this is a key point that why our Mixed-ISC can detect weak signal networks.
  \item In the MR step, for the feasibility of theoretical analysis for Mixed-ISC, we compute $\hat{\pi}_{i}$ by applying $\hat{Y}_{i}$ dividing its $l_{2}$-norm in the theoretical analysis part, since $l_{2}$-norm is easier to be analysized than $l_{1}$-norm for mixed membership community detection problem, just as that in \cite{OCCAM}.
\end{itemize}

\section{Main results}\label{sec5}
In this section, we show asymptotic consistency of Mixed-ISC for fitting the DCMM. Let $\mathscr{D}_{\tau}=\mathscr{D}+\tau I$, where $\mathscr{D}$ is an $n\times n$ diagonal matrix whose $i$-th diagonal entry is $\mathscr{D}(i,i)=\sum_{j=1}^{n}\Omega(i,j)$. The population Laplacian matrix with regularization is defined as
\begin{align*}
\mathscr{L}_{\tau}=\mathscr{D}^{-1/2}_{\tau}\Omega\mathscr{D}^{-1/2}_{\tau}.
\end{align*}
Let $\theta_{\mathrm{max}}=\mathrm{max}_{1\leq i\leq n}\{\theta(i)\}, \theta_{\mathrm{min}}=\mathrm{min}_{1\leq i\leq n}\{\theta(i)\}$, and let $\mathcal{N}_{k}=\{1\leq i\leq n: \pi_{i}(k)=1\}$ be the set of pure nodes of community $k, 1\leq k\leq K$. First, we assume that there are three constants $c_{1}, c_{2}\in (0,1)$ and $c_{3}>0$ such that
\begin{align}\label{a1}
\underset{1\leq k\leq K}{\mathrm{min}}|\mathcal{N}_{k}|\geq c_{1}n,~~\underset{1\leq k\leq K}{\mathrm{min}}\sum_{i\in \mathcal{N}_{k}}\theta^{2}(i)\geq c_{2}\|\theta\|^{2},~~\theta_{\mathrm{max}}\leq c_{3}.
\end{align}
Second, assume that as $n\rightarrow \infty$,
\begin{align}\label{a2}
\frac{\mathrm{log}^{2}(n)\sqrt{\theta_{\mathrm{max}}\|\theta\|_{1}}}{\theta_{\mathrm{min}}\|\theta\|\sqrt{n}}\rightarrow 0.
\end{align}
For convenience, set $err_{n}$ as
\begin{align*}
err_{n}=&(\frac{\Delta_{\mathrm{max}}}{\sqrt{\tau+\Delta_{\mathrm{max}}}}+\frac{\delta_{\mathrm{max}}}{\sqrt{\tau+\delta_{\mathrm{max}}}})\frac{\mathrm{max}\{1-\sqrt{\frac{\tau+\Delta_{\mathrm{min}}}{\tau+\delta_{\mathrm{max}}}}, \sqrt{\frac{\tau+\Delta_{\mathrm{max}}}{\tau+\delta_{\mathrm{min}}}}-1\}}{\sqrt{\tau+\Delta_{\mathrm{min}}}}\\
&+\frac{C\sqrt{\mathrm{log}(n)\theta_{\mathrm{max}}\|\theta\|_{1}}}{\sqrt{(\tau+\Delta_{\mathrm{min}})(\tau+\delta_{\mathrm{min}})}},
\end{align*}
where $\Delta_{\mathrm{max}}=d_{\mathrm{max}},\Delta_{\mathrm{min}}=d_{\mathrm{min}}, \delta_{\mathrm{max}}=\underset{1\leq i\leq n}{\mathrm{max}}\mathscr{D}(i,i), \delta_{\mathrm{min}}=\underset{1\leq i\leq n}{\mathrm{min}}\mathscr{D}(i,i)$. Next lemma bounds $\|L_{\tau}-\mathscr{L}_{\tau}\|$.
\begin{lem}\label{boundL}
Under $DCMM(n, P, \Theta, \Pi)$, if assumptions (\ref{a1})-( \ref{a2}) hold, with probability at least $1-o(n^{-3})$, we have
\begin{align*}
\|L_{\tau}-\mathscr{L}_{\tau}\|\leq err_{n}.
\end{align*}
\end{lem}
By basic algebra, we have the rank of $\mathscr{L}_{\tau}$ is $K$, hence $\mathscr{L}_{\tau}$ has $K$ nonzero eigenvalues, let $\{\lambda_{i},\eta_{i}\}_{i=1}^{K}$ be such leading $K$ eigenvalues and their respective eigenvectors with unit-norm. Define the $n\times (K+1)$ matrix $X$ as a population version of $\hat{X}$ as $X=[\eta_{1},\eta_{2},\ldots, \eta_{K}, \textbf{0}]\cdot\mathrm{diag}(\lambda_{1},\lambda_{2},\ldots,\lambda_{K},0)$,
where $\textbf{0}$ is an $n\times 1$ vector with elements being zeros. Let $X^{*}$ be the row-normalization version of $X$, i.e., $X^{*}$ is obtained be normalizing each rows of $X$ to have unit-length. To obtain the bound of $\|\hat{X}^{*}-X^{*}\|_{F}$, the eigenvalues of $L_{\tau}$ and $\mathscr{L}_{\tau}$ should satisfy the following assumption: For $K+1\leq i\leq n$, we have
\begin{align}\label{a3}
\hat{\lambda}_{1}\geq \ldots \geq \hat{\lambda}_{K}>0, \lambda_{1}\geq \ldots \geq \lambda_{K}>0, \mathrm{and~} \hat{\lambda}_{K}>|\hat{\lambda}_{i}|.
\end{align}
Lemma \ref{boundXstar} provides the bound of $\|\hat{X}^{*}-X^{*}\|_{F}$, which is the corner stone to characterize the behavior of our Mixed-ISC approach.
\begin{lem}\label{boundXstar}
Under $DCMM(n, P, \Theta, \Pi)$, set $m=\mathrm{min}_{i}\{\mathrm{min}\{\|\hat{X}_{i}\|, \|X_{i}\|\}\}$ as the length of the shortest row in $\hat{X}$ and $X$. If assumptions (\ref{a1})-( \ref{a3}) hold, with probability at least $1-o(n^{-3})$, we have
\begin{align*}
\|\hat{X}^{*}-X^{*}\|_{F}\leq \frac{1}{m}(\sqrt{K+1}+\frac{K\delta_{\mathrm{max}}}{\lambda_{K}(\tau+\delta_{\mathrm{max}})})err_{n}.
\end{align*}
\end{lem}
Perform K-medians clustering on the rows of $X^{*}$ and obtain the $K$ cluster centers $v_{1}, v_{2},\ldots, v_{K}\in\mathcal{R}^{1\times (K+1)}$. Let $V$ be the $K\times (K+1)$ matrix such that the $i$-th row of $V$ is $v_{i}$ for $1\leq i\leq K$. To obtain the bound of $\|\hat{V}-V\|_{F}$, we need some extra conditions. Similar as \cite{OCCAM}, we define a Hausdorff distance which is used to measure the dissimilarity between two cluster centers as $D_{H}(S, T)=\mathrm{min}_{\sigma\in \Sigma}\|S-T\sigma\|_{F}$ for any $K\times (K+1)$ matrix $S$ and $T$ where $\Sigma$ is a set of $(K+1)\times (K+1)$ permutation matrix. The sample loss function for K-medians is defined by
\begin{align*}
\mathcal{L}_{n}(Q;S)=\frac{1}{n}\sum_{i=1}^{n}\mathrm{min}_{1\leq k\leq K}\|Q_{i}-S_{k}\|_{F},
\end{align*}
where $Q\in \mathcal{R}^{n\times (K+1)}$ is a matrix whose rows $Q_{i}$ are vectors to be clustered, and $S\in \mathcal{R}^{K\times (K+1)}$ is a matrix whose rows $S_{k}$ are cluster centers. Assuming the rows of $Q$ are i.i.d. random vectors sampled from a distribution $\mathcal{G}$, we similarly define the population loss function for K-medians by
\begin{align*}
\mathcal{L}(\mathcal{G};S)=\int \mathrm{min}_{1\leq k\leq K}\|x-S_{k}\|_{F}d\mathcal{G}.
\end{align*}
Let $\mathcal{F}$ be the distribution of $X^{*}_{i}$,  assume the following conditions on $\mathcal{F}$ hold:
\begin{align}\label{a4}
&\mathrm{~~~~~~Let}~ V_{\mathcal{F}}=\mathrm{argmin}_{U}\mathcal{L}(\mathcal{F};U) \mathrm{~be ~the~ global ~minimizer ~of ~the ~}\\
&~~~~~~\mathrm{population~loss ~ function~} \mathcal{L}(\mathcal{F};U). \mathrm{~Then~} V_{\mathcal{F}}=V\mathrm{~up~ to~ a ~row }\notag\\
&~~~~~~\mathrm{permutation.~Further,~ there~exists~a~global~constant~}\kappa~ \mathrm{such} \notag\\
&~~~~~~\mathrm{that~}\kappa K^{-1}D_{H}(U,V_{\mathcal{F}})\leq\mathcal{L}(\mathcal{F};U)-\mathcal{L}(\mathcal{F};V_{\mathcal{F}})~\mathrm{~for~all~}U.\notag
\end{align}
Condition (\ref{a4}) essentially states that the population K-medians loss function, which is determined by $\mathcal{F}$, has a unique minimum at the right place.

Next lemma bounds $\|\hat{V}-V\|_{F}$.
\begin{lem}\label{boundV}
Under $DCMM(n,P,\Theta,\Pi)$, if assumptions (\ref{a1})-(\ref{a4}) hold, then with probability at least $\mathrm{P}(n, K)$, we have
\begin{align*}
\|\hat{V}- V\|_{F}\leq C\frac{K}{m\sqrt{n}}(\sqrt{K+1}+\frac{K\delta_{\mathrm{max}}}{\lambda_{K}(\tau+\delta_{\mathrm{max}})})err_{n},
\end{align*}
\end{lem}
where $C$ is a constant and $\mathrm{P}(n,K)\rightarrow 1$ as $n\rightarrow \infty$.

Assume that there exists a global constant $m_{V}>0$ such that
\begin{align}\label{a5}
\lambda_{\mathrm{min}}(VV')\geq m_{V}.
\end{align}
For convenience, set
\begin{align*}
Err_{n}&=(\frac{2K^{2}}{m_{V}-2\sqrt{K}\frac{CK}{m\sqrt{n}}(\sqrt{K+1}+\frac{K\delta_{\mathrm{max}}}{\lambda_{K}(\tau+\delta_{\mathrm{max}})})err_{n}}+\sqrt{K})\frac{CK(\sqrt{K+1}+\frac{K\delta_{\mathrm{max}}}{\lambda_{K}(\tau+\delta_{\mathrm{max}})})err_{n}}{mm_{V}}\\
&~~~+\frac{K}{mm_{V}}(\sqrt{K+1}+\frac{K\delta_{\mathrm{max}}}{\lambda_{K}(\tau+\delta_{\mathrm{max}})})err_{n}.
\end{align*}
Next lemma bounds $\|\hat{Y}-Y\|_{F}$ under mild conditions, where $Y$ is defined as $Y=\mathrm{max}(X^{*}V'(VV')^{-1},0)$.
\begin{lem}\label{boundY}
Under $DCMM(n,P,\Theta,\Pi)$, if assumptions (\ref{a1})-(\ref{a5}) hold, then with probability at least $\mathrm{P}(n,K)$, we have
\begin{align*}
\|\hat{Y}-Y\|_{F}\leq Err_{n}.
\end{align*}
\end{lem}
After obtain the bounds for $\|\hat{X}^{*}-X^{*}\|_{F}$, $\|\hat{V}-V\|_{F}$ and $\|\hat{Y}-Y\|_{F}$, we can give  a theoretical bound on $\|\hat{\Pi}-\Pi\|_{F}$ which is the main theoretical result of this paper.
\begin{thm}\label{main}
Under $DCMM(n,P,\Theta,\Pi)$, set $m_{Y}=\mathrm{min}_{i}\{\|Y_{i}\|\}$ as the length of the shortest row in $Y$. If assumptions (\ref{a1})-(\ref{a5}) hold, then with probability at least $\mathrm{P}(n,K)$, we have
\begin{align*}
\frac{\|\hat{\Pi}-\Pi\|_{F}}{\sqrt{n}}\leq \frac{2Err_{n}}{m_{Y}\sqrt{n}}.
\end{align*}
\end{thm}
\begin{rem}
Assumptions (\ref{a1}) and (\ref{a2}) are the same as that of Lemma 3.2 in \cite{mixedSCORE} since we need to apply the bound of $\|A-\Omega\|$ given by \cite{mixedSCORE} to bound $\|L_{\tau}-\mathscr{L}_{\tau}\|$. It is challenging to estimate the constants $m, m_{V}, m_{Y},\lambda_{K}$ at present, and we leave it as future work.
\end{rem}

\section{Simulations}\label{sec6}
We investigate the performance of our Mixed-ISC by some artificial networks. We compare it with Mixed-SCORE, OCCAM and GeoNMF. Based on the principle that clustering errors under any kind of measurements should not depend on how we label the $K$ communities, we need to consider the permutation of cluster label when measuring the clustering errors. Therefore, we measure the performance of mixed membership community detection method by the mixed-Hamming error rate which is defined as
\begin{align*}
\mathrm{min}_{O\in\{ K\times K\mathrm{permutation~matrix}\}}\frac{1}{n}\|\hat{\Pi}O-\Pi\|_{1},
\end{align*}
where $\Pi$ and $\hat{\Pi}$ are the true and estimated mixed membership matrices respectively. For simplicity, we write the mixed-Hamming error rate as $\sum_{i=1}^{n}\|\hat{\pi}_{i}-\pi_{i}\|_{1}/n$. %For all experiments in this section, we report the mean of the mixed-Hamming error rates for every approach.% therefore in all figures in this section, the y-axis indicates the mean of $\sum_{i=1}^{n}\|\hat{\pi}_{i}-\pi_{i}\|_{1}/n$.

Unless specified, for all experiments below, we set $n=500$ and $K=3$. For $0\leq n_{0}\leq 160$, let each block own $n_{0}$ number of pure nodes. For the top $3n_{0}$ nodes $\{1,2, \ldots, 3n_{0}\}$, we let these nodes be pure and let nodes $\{3n_{0}+1, 3n_{0}+2,\ldots, 500\}$ be mixed. Fixing $x\in [0, \dfrac{1}{2})$, let all the mixed nodes have four different memberships $(x, x, 1-2x), (x, 1-2x, x), (1-2x, x, x)$ and $(1/3,1/3,1/3)$, each with $\dfrac{500-3n_{0}}{4}$ number of nodes. Fixing $\rho\in(0, 1)$, the mixing matrix $P$ has diagonals 0.8 and off-diagonals $\rho$. Fixing $z\geq 1$, we generate the degree parameters such that $1/\theta(i)\overset{iid}{\sim}U(1,z)$, where $U(1,z)$ denotes the uniform distribution on $[1, z]$. For  each parameter setting, we report the averaged mixed-Hamming error rate over 50 repetitions. The details of designed experiments are described as follows:

\texttt{Experiment 1: Fraction of pure nodes.} Fix $(x,\rho, z)=(0.4, 0.3, 4)$ and let $n_{0}$ range in $\{40, 60, 80, 100, 120, 140, 160\}$. A larger $n_{0}$ indicates a case with higher fraction of pure nodes. The numerical results are shown in the top left panel of Figure \ref{EX}, from which we can find that all methods perform poor when the fraction of pure nodes is small. For a larger $n_{0}$, Mixed-ISC performs similar as Mixed-SCORE and both two algorithms outperform OCCAM and GeoNMF.

\texttt{Experiment 2: Connectivity across communities.} Fix $(x,n_{0},z)=(0.4, 100,4)$ and let $\rho$ range in $\{0, 0.05, 0.1, \ldots, 0.4\}$. A lager $\rho$ generate more edges across different communities (hence a dense network). The results are displayed in the top right panel of Figure \ref{EX}. We can find that all methods perform poorer as $\rho$ increases, this phenomenon occurs due to the fact that more edges across different communities lead to a case that different communities tend to be into a giant community and hence a case that is more challenging to detect for any algorithms. Meanwhile, generally speaking, the results show that our Mixed-ISC always outperforms its competitors in this experiment.

\begin{figure}
\centering
\subfigure[Experiment 1]{\includegraphics[width=0.36\textwidth]{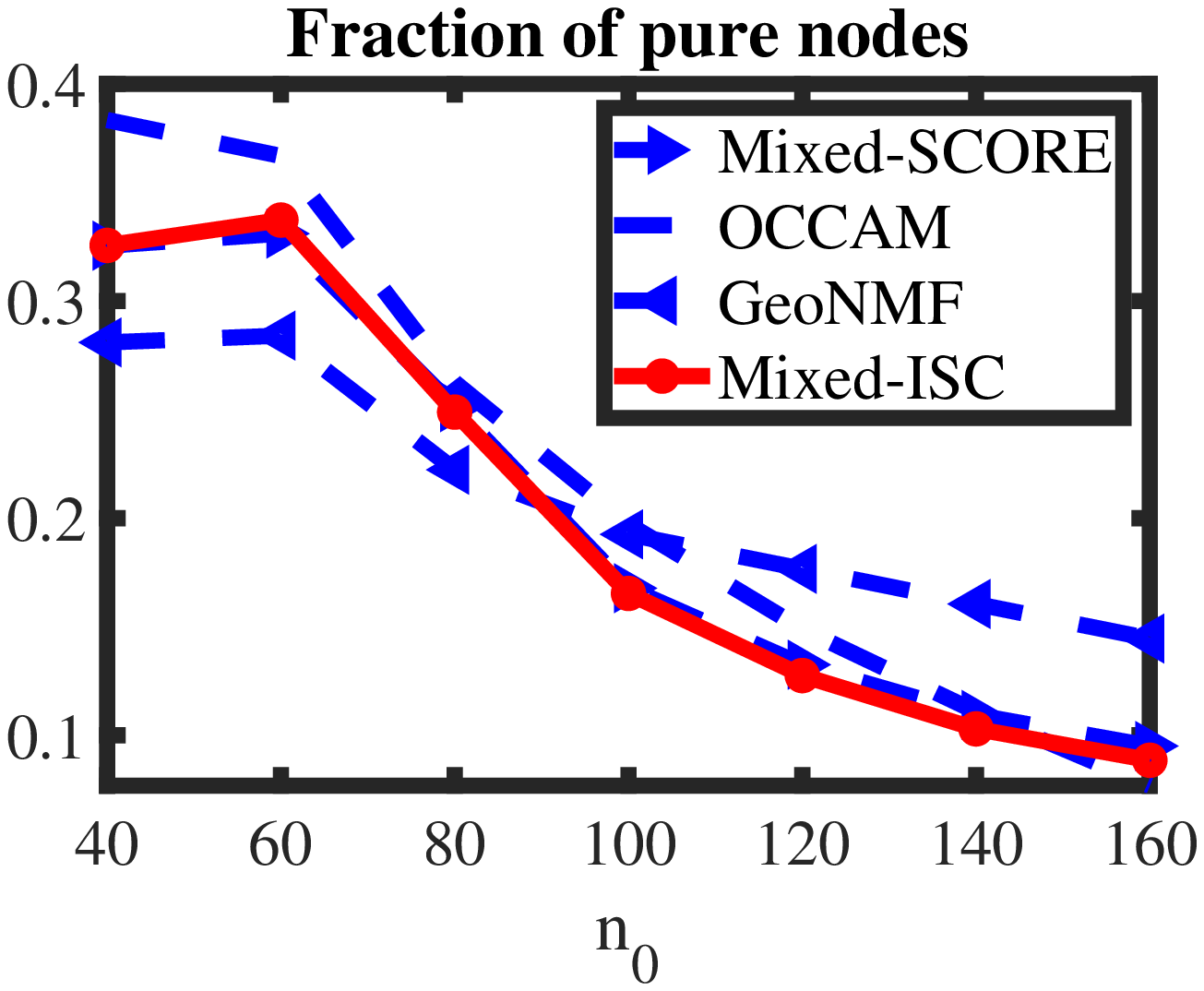}}
\subfigure[Experiment 2]{\includegraphics[width=0.36\textwidth]{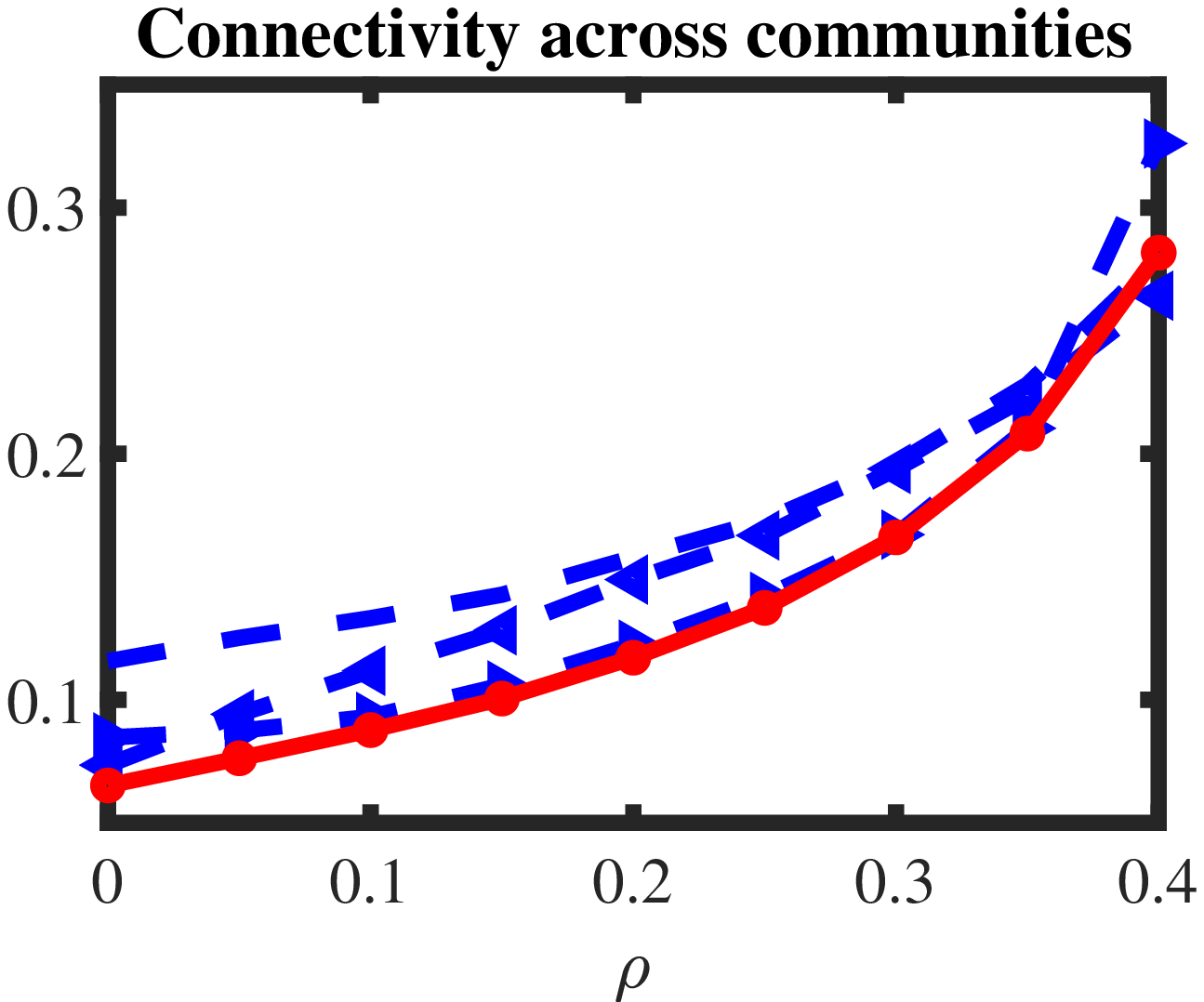}}
\subfigure[Experiment 3]{\includegraphics[width=0.36\textwidth]{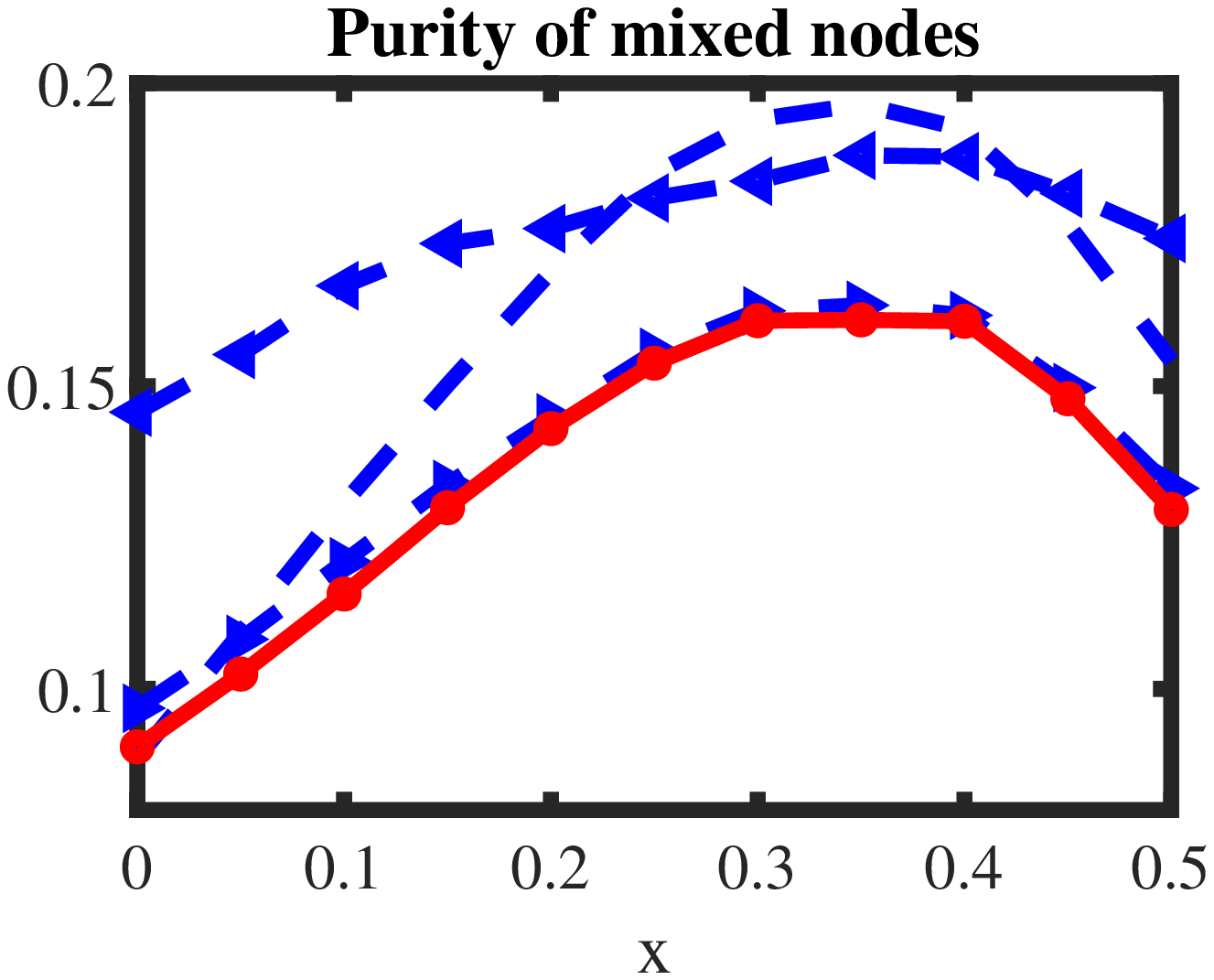}}
\subfigure[Experiment 4]{\includegraphics[width=0.36\textwidth]{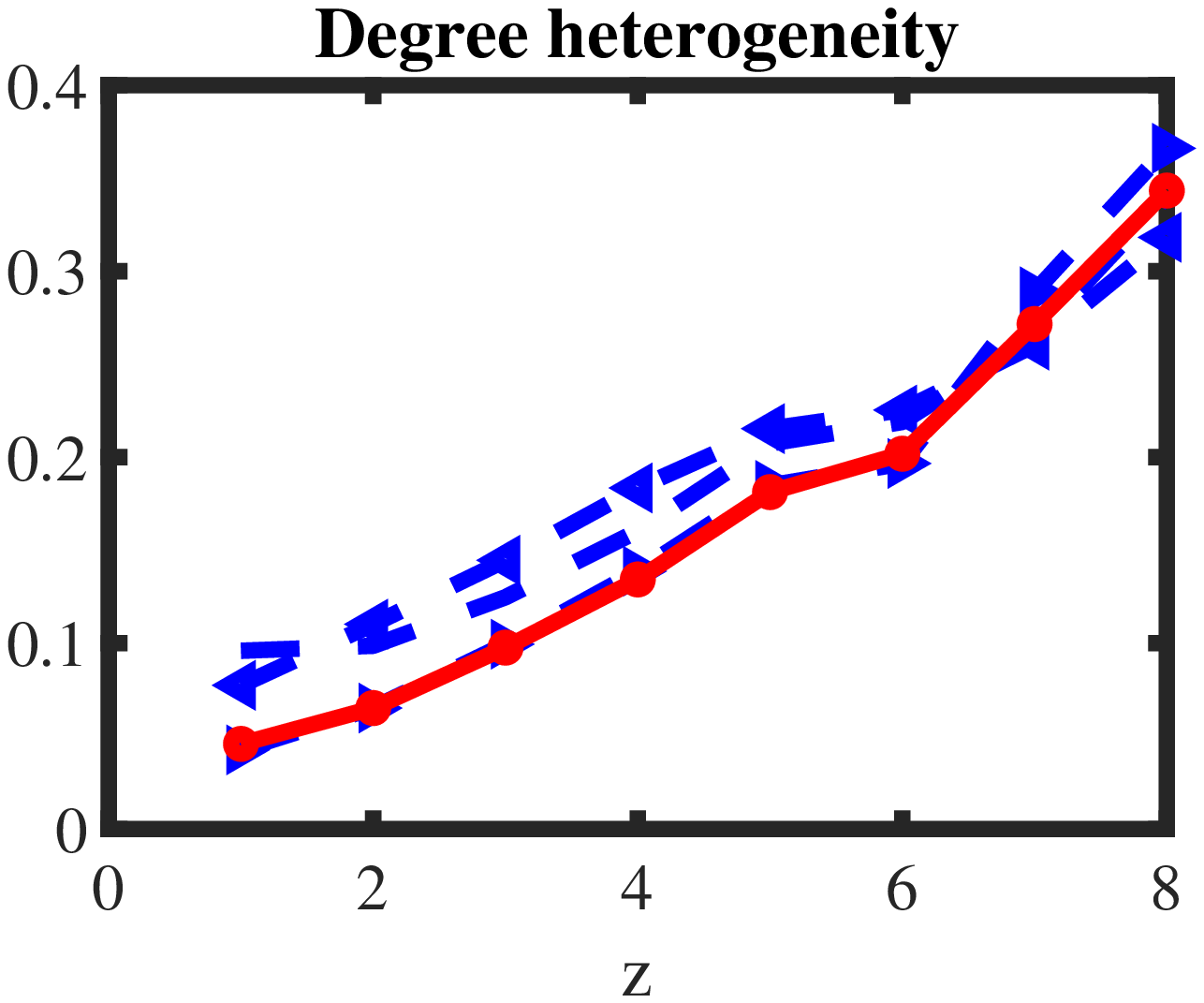}}
\caption{Estimation errors of Experiments 1-4 (y-axis: the mean of $\sum_{i=1}^{n}n^{-1}\|\hat{\pi}_{i}-\pi_{i}\|_{1}$).}
\label{EX}
\end{figure}
\texttt{Experiment 3: Purity of mixed nodes.} Fix $(n_{0},\rho, z)=(100, 0.3, 4)$, and let $x$ range in $\{0, 0.05, \ldots, 0.5\}$. As $x$ increases to 1/3, these mixed nodes become less pure and they become more pure as $x$ increases further. The bottom left panel of Figure \ref{EX} records the numerical results of this experiment. It is obvious that Mixed-ISC shares similar performances with Mixed-SCORE and both two methods significantly outperform OCCAM and GeoNMF in this experiment.

\texttt{Experiment 4: Degree heterogeneity.} Fix $(n_{0},\rho, x)=(100, 0.3, 0.4)$, and let $z$ range in $\{1,2,\ldots,8\}$. A larger $z$ gives smaller $\theta(i)$ for any node $i$, hence a more heterogeneous case and fewer edges generated (a case that is more challenging to detect for any algorithms). The bottom right panel of Figure \ref{EX} presents the numerical results. We see that all methods perform poorer as $z$ increases. Mixed-ISC and Mixed-SCORE have competitive perfromances and both two methods enjoy better performances than that of OCCAM and GeoNMF for a small $z$. However, when $z$ is too large (see, $z\geq 6$), all methods perform similar with high mixed-Hamming error rates.
%%%
\section{Application to empirical datasets}\label{sec7}
%The performances of Mixed-ISC and its competitors on the two weak signal networks Simmons and Caltech for community detection are studied in Table \ref{real2errors}, which tells us that our Mixed-ISC significantly outperforms its competitors on this two networks.
In this section, we study the performances of Mixed-ISC as well as its competitors on SNAP ego-networks and the Coauthorship network.
\subsection{Application to SNAP ego-networks} The ego-networks dataset can be found in \href{http}{http://snap.stanford.edu/data/}. Since the mixed memberships of SNAP ego-networks are known \cite{leskovec2016snap,SNAPego,mcauley2014discovering,OCCAM}, we use the SNAP ego-networks to investigate the performances of our Mixed-ISC via comparing it with Mixed-SCORE, OCAM and GeoNMF by reporting the mixed-Hamming error rate. To get a better sense of what the different social networks look like and how different characteristics potentially affect performance of our Mixed-ISC, we report the following summary statistics for each network: (1) number of nodes $n$ and number of communities $K$. (2) average node degree $\bar{d}$. (3) density $\sum_{i,j}A(i,j)/(n(n-1))$, i.e., the overall edge probability. (4) the proportion of overlapping nodes $r_{o}$, i.e., $r_{o}=\frac{\mathrm{number~of~nodes~with~mixed~membership}}{n}$. We report the means and standard deviations of these measures for each of the social networks in Table \ref{dataSNAP}.

From Table \ref{dataSNAP}, we see that Facebook and GooglePlus networks tend to be lager than Twitter networks, while Twitter networks are denser with larger density. Meanwhile, the proportions of overlapping nodes in Twitter networks tend to be larger than that of Facebook and GooglePlus networks. Compared with the ego-networks applied in \cite{OCCAM}, the parsed ego-networks used in this paper have more networks to be analyzed. %\textbf{Meanwhile, after obtaining the membership matrix $\Pi$ of each ego-network, to compute the mixed-Hamming error rate, we take a row-normalization step such that the sum of each row of $\Pi$ equals to 1.}
\begin{table}[h!]
%\footnotesize
\centering
\caption{Mean (SD) of summary statistics for ego-networks.}
\label{dataSNAP}
%\resizebox{\columnwidth}{!}{
\begin{tabular}{cccccccccc}
\toprule
&\#Networks&$n$&$K$&$\bar{d}$&Density&$r_{o}$\\
\midrule
Facebook&7&236.57&3&30.61&0.15&0.009\\
&-&(228.53)&(1.15)&(29.41)&(0.058)&(0.008)\\
\hline
GooglePlus&58&433.22&2.22&66.81&0.18&0.005\\
&-&(327.70)&(0.46)&(65.2)&(0.11)&(0.005)\\
\hline
Twitter&255&60.64&2.63&17.87&0.33&0.02\\
&-&(30.77)&(0.83)&(9.97)&(0.17)&(0.008)\\
\bottomrule
\end{tabular}%}
\end{table}

We also report these  statistics for strong signal networks and weak signal networks of SNAP ego-networks in Table \ref{dataSNAPstrongA} and Table \ref{dataSNAPweakA}, respectively, where we use $\mathrm{Facebook}_{\mathrm{strong}}$ to denote strong signal networks in Facebook (similar notations hold for other networks) for convenience. From Table \ref{dataSNAPstrongA} and Table \ref{dataSNAPweakA}, we see that weak signal networks and strong signal networks share similar values of $n, \bar{d}, \mathrm{Density}$ and $r_{o}$ while weak signal networks tend to have larger $K$ than that of strong signal networks.
\begin{table}[h!]
%\footnotesize
\centering
\caption{Mean (SD) of summary statistics for strong signal ego-networks.}
\label{dataSNAPstrongA}
%\resizebox{\columnwidth}{!}{
\begin{tabular}{cccccccccc}
\toprule
&\#Networks&$n$&$K$&$\bar{d}$&Density&$r_{o}$\\
\midrule
$\mathrm{Facebook}_{\mathrm{strong}}$&7&236.57&3&30.61&0.15&0.009\\
&-&(228.53)&(1.15)&(29.41)&(0.058)&(0.008)\\
\hline
$\mathrm{GoolePlus}_{\mathrm{strong}}$&45&433.36&2.20&63.14&0.17&0.005\\
&-&(317.08)&(0.40)&(49.45)&(0.10)&(0.005)\\
\hline
$\mathrm{Twitter}_{\mathrm{strong}}$&156&61.42&2.47&18.50&0.33&0.02\\
&-&(31.89)&(0.73)&(10.58)&(0.15)&(0.008)\\
\bottomrule
\end{tabular}%}
\end{table}
\begin{table}[h!]
%\footnotesize
\centering
\caption{Mean (SD) of summary statistics for weak signal ego-networks.}
\label{dataSNAPweakA}
%\resizebox{\columnwidth}{!}{
\begin{tabular}{cccccccccc}
\toprule
&\#Networks&$n$&$K$&$\bar{d}$&Density&$r_{o}$\\
\midrule
$\mathrm{Facebook}_{\mathrm{weak}}$&0&-&-&-&-&-\\
&-&-&-&-&-&-\\
\hline
$\mathrm{GoolePlus}_{\mathrm{weak}}$&13&436.23&2.31&79.52&0.22&0.005\\
&-&(376.06)&(0.63)&(104.89)&(0.13)&(0.004)\\
\hline
$\mathrm{Twitter}_{\mathrm{weak}}$&99&59.41&2.87&16.87&0.33&0.02\\
&-&(29.03)&(0.91)&(8.90)&(0.19)&(0.007)\\
\bottomrule
\end{tabular}%}
\end{table}

To compare methods, we report mean and the corresponding standard deviation of mixed-Hamming error rates for each platforms in Table \ref{ErrorSNAPstrong} and Table \ref{ErrorSNAPweak}. For Facebook networks, Mixed-SCORE performs slightly better than Mixed-ISC, GeoNMF and OCCAM. For both strong and weak signal networks in GooglePlus and Twitter, Mixed-ISC outperforms other methods, while OCCAM and GeoNMF have similar results and are better than Mixed-SCORE. We also can find that the results in Table \ref{ErrorSNAPstrong} are smaller than that in Table \ref{ErrorSNAPweak} for these four methods in GooglePlus networks and Twitter networks. We may argue that weak signal networks may contain more noise and are much harder to detect communities than strong signal networks. Therefore, it is worth to study the problem of community detection for weak signal networks and strong signal networks separately.

\begin{table}[h!]
\centering
\caption{Mean (SD) of mixed-Hamming error rates for strong signal ego-networks.}
\label{ErrorSNAPstrong}
%\resizebox{\columnwidth}{!}{
\begin{tabular}{cccccccccc}
\toprule
&$\mathrm{Facebook}_{\mathrm{strong}}$&$\mathrm{GooglePlus}_{\mathrm{strong}}$&$\mathrm{Twitter}_{\mathrm{strong}}$\\
\midrule
Mixed-SCORE&\textbf{0.2496}(0.1322)&0.3703(0.1088)&0.2905(0.1406)\\
OCCAM&0.2610(0.1367)&0.3510(0.1276)&0.2709(0.1508)\\
GeoNMF&0.2584(0.1262)&0.3435(0.1115)&0.2712(0.1375)\\
\hline
Mixed-ISC&0.2582(0.1295)&\textbf{0.3117}(0.1330)&\textbf{0.2414}(0.1446)\\
\bottomrule
\end{tabular}%}
\end{table}
\begin{table}[h!]
\centering
\caption{Mean (SD) of mixed-Hamming error rates for weak signal ego-networks.}
\label{ErrorSNAPweak}
%\resizebox{\columnwidth}{!}{
\begin{tabular}{cccccccccc}
\toprule
&$\mathrm{Facebook}_{\mathrm{weak}}$&$\mathrm{GooglePlus}_{\mathrm{weak}}$&$\mathrm{Twitter}_{\mathrm{weak}}$\\
\midrule
Mixed-SCORE&-&0.3983(0.0927)&0.3376(0.1047)\\
OCCAM&-&0.3751(0.0971)&0.3102(0.1182)\\
GeoNMF&-&0.3770(0.0909)&0.3075(0.1107)\\
\hline
Mixed-ISC&-&\textbf{0.3527}(0.1072)&\textbf{0.2847}(0.1172)\\
\bottomrule
\end{tabular}%}
\end{table}

Furthermore, similar as that in \cite{SCORE+}, the weak signal networks can also be defined based on $A$'s variant, the regularized Laplacian matrix $L_{\tau}$, as follows:
\begin{defin}\label{weakL}
Given an un-weighted, undirected, no-loops and connected network $\mathcal{N}$ with $K$ communities. The network is a weak signal network if $1-|\frac{\hat{\lambda}_{K+1}}{\hat{\lambda}_{K}}|\leq 0.1$, otherwise it is strong signal network, where $\hat{\lambda}_{k}$ denotes the $k$-the leading eigenvalue of $L_{\tau}, k=K, K+1$.
\end{defin}
Follow this definition, numbers of weak signal networks of the SNAP ego-networks are slightly different as that based on $A$, see Table \ref{SNAPweakL}. Based on the regularized Laplacian matrix, the number of weak signal networks in Facebook networks and GooglePlus networks  increase, while it decreases for Twitter networks. Meanwhile, we also report the summary statistics for strong signal networks and weak signal networks (defined based on $L_{\tau}$) in Table \ref{dataSNAPstrongLtau} and Table \ref{dataSNAPweakLtau}. Compared with the summaries in Table  \ref{dataSNAPstrongA} and \ref{dataSNAPweakA}, the properties of these networks are slightly different.  % where this two tables show that weak signal network tends to have a larger $K$ than that of strong signal network.
\begin{table}[h!]
\centering
\caption{Numbers of weak signal and strong signal networks (defined based on $L_{\tau}$) in SNAP ego-networks.}
\label{SNAPweakL}
%\resizebox{\linewidth}{!}{
\begin{tabular}{cccccccccc}
\toprule
&\#Networks&\# Weak Signal networks&\# Strong Signal networks\\
\midrule
Facebook&7&1&6\\
\hline
GooglePlus&58&16&42\\
\hline
Twitter&255&96&159\\
\bottomrule
\end{tabular}%}
\end{table}

\begin{table}[h!]
%\footnotesize
\centering
\caption{Mean (SD) of summary statistics for strong signal (defined based on $L_{\tau}$) ego-networks.}
\label{dataSNAPstrongLtau}
%\resizebox{\columnwidth}{!}{
\begin{tabular}{cccccccccc}
\toprule
&\#Networks&$n$&$K$&$\bar{d}$&Density&$r_{o}$\\
\midrule
$\mathrm{Facebook}_{\mathrm{strong}}$&6&172.50&2.67&29.06&0.17&0.011\\
&-&(167.89)&(0.82)&(31.91)&(0.048)&(0.008)\\
\hline
$\mathrm{GoolePlus}_{\mathrm{strong}}$&42&428.88&2.21&63.32&0.18&0.005\\
&-&(337.00)&(0.47)&(50.81)&(0.10)&(0.005)\\
\hline
$\mathrm{Twitter}_{\mathrm{strong}}$&159&58.01&2.44&18.07&0.34&0.02\\
&-&(26.81)&(0.67)&(10.09)&(0.15)&(0.007)\\
\bottomrule
\end{tabular}%}
\end{table}
\begin{table}[h!]
%\footnotesize
\centering
\caption{Mean (SD) of summary statistics for weak signal (defined based on $L_{\tau}$) ego-networks.}
\label{dataSNAPweakLtau}
%\resizebox{\columnwidth}{!}{
\begin{tabular}{cccccccccc}
\toprule
&\#Networks&$n$&$K$&$\bar{d}$&Density&$r_{o}$\\
\midrule
$\mathrm{Facebook}_{\mathrm{weak}}$&1&621&5&39.93&0.06&0.002\\
&-&(0)&(0)&(0)&(0)&(0)\\
\hline
$\mathrm{GoolePlus}_{\mathrm{weak}}$&16&444.63&2.25&75.97&0.19&0.005\\
&-&(312.17)&(0.48)&(94.73)&(0.13)&(0.005)\\
\hline
$\mathrm{Twitter}_{\mathrm{weak}}$&96&65.00&2.94&17.53&0.31&0.02\\
&-&(36.13)&(0.96)&(9.82)&(0.18)&(0.008)\\
\bottomrule
\end{tabular}%}
\end{table}

Table \ref{ErrorSNAPstrongL} and Table \ref{ErrorSNAPweakL} report the numerical results of Mixed-ISC and its three competitors. Except the strong signal networks of the Facebook datasets, Mixed-ISC always significantly outperforms other methods on both strong and weak signal networks of GooglePlus and Twitter datasets. Furthermore, Mixed-ISC performs best on the weak signal network of Facebook datasets. This results are slightly different as that in Table \ref{ErrorSNAPstrong} and Table \ref{ErrorSNAPweak}. It is very interesting to find that, in the mass, the mixed Humming error rates in Table \ref{ErrorSNAPstrongL} are smaller than the results in Table \ref{ErrorSNAPstrong}, but the error rates in Table \ref{ErrorSNAPweakL} are larger than the results in Table \ref{ErrorSNAPweak}.  Perhaps the regularized Laplacian matrix $L_{\tau}$ is more sensitive to identify weak signal networks than the adjacency matrix $A$.
\begin{table}[h!]
\centering
\caption{Mean (SD) of mixed-Hamming error rates for strong signal (defined based on $L_{\tau}$) ego-networks.}
\label{ErrorSNAPstrongL}
%\resizebox{\columnwidth}{!}{
\begin{tabular}{cccccccccc}
\toprule
&$\mathrm{Facebook}_{\mathrm{strong}}$&$\mathrm{GooglePlus}_{\mathrm{strong}}$&$\mathrm{Twitter}_{\mathrm{strong}}$\\
\midrule
Mixed-SCORE&\textbf{0.2339}(0.1374)&0.3675(0.1095)&0.2753(0.1424)\\
OCCAM&0.2530(0.1479)&0.3420(0.1206)&0.2507(0.1515)\\
GeoNMF&0.2460(0.1368)&0.3400(0.1105)&0.2558(0.1412)\\
\hline
Mixed-ISC&0.2549(0.1415)&\textbf{0.3062}(0.1321)&\textbf{0.2215}(0.1453)\\
\bottomrule
\end{tabular}%}
\end{table}
\begin{table}[h!]
\centering
\caption{Mean (SD) of mixed-Hamming error rates for weak signal (defined based on $L_{\tau}$) ego-networks.}
\label{ErrorSNAPweakL}
%\resizebox{\columnwidth}{!}{
\begin{tabular}{cccccccccc}
\toprule
&$\mathrm{Facebook}_{\mathrm{weak}}$&$\mathrm{GooglePlus}_{\mathrm{weak}}$&$\mathrm{Twitter}_{\mathrm{weak}}$\\
\midrule
Mixed-SCORE&0.3441&0.4003(0.0922)&0.3636(0.0788)\\
OCCAM&0.3092&0.3940(0.1176)&0.3454(0.0937)\\
GeoNMF&0.2939&0.3854(0.0949)&0.3376(0.0841)\\
\hline
Mixed-ISC&\textbf{0.2782}&\textbf{0.3595}(0.1110)&\textbf{0.3190}(0.0914)\\
\bottomrule
\end{tabular}%}
\end{table}

%%%%
\subsection{Application to Coauthorship network}
Since there is no ground truth of the nodes membership for the Coauthorship network \cite{ji2016coauthorship, mixedSCORE}, similar as that in \cite{mixedSCORE}, we only provide the estimated PMFs of the ``Carroll-Hall'' community \footnote{The respective estimated PMF of the ``North Carolina'' community for an author just equals 1 minus the author's weight of the ``Carroll-Hall'' community.} for 20 authors, where the 20 authors are also studied in Table 2 in \cite{mixedSCORE} and 19 of them (except Jiashun Jin) are regarded with highly mixed memberships in \cite{mixedSCORE}.  The results are in Table \ref{Coauthorship}.
\begin{table}
\centering
\caption{Estimated PMF of the ``Carroll-Hall'' community for the Coauthorship network.}
\label{Coauthorship}
%\resizebox{\columnwidth}{!}{
\begin{tabular}{lccccccccccccc}
\hline
Methods &Mixed-ISC&Mixed-SCORE&OCCAM&GeoNMF\\
\hline
Jianqing Fan&87.36\%&56.21\% &65.51\%&78.01\%\\
Jason P Fine&97.16\%&56.79\% &65.15\%&78.21\%\\
Michael R Kosorok&95.13\%&62.45\% &61.55\%&80.12\%\\
J S Marron&100.00\%&41.00\% &74.06\%&72.48\%\\
Hao Helen Zhang&100.00\%&48.45\% &70.05\%&75.27\%\\
Yufeng Liu&100.00\%&46.03\% &71.39\%&74.38\%\\
Xiaotong Shen&100.00\%&41.00\% &74.06\%&72.48\%\\
Kung-Sik Chan&98.70\%&42.33\% &73.37\%&72.99\%\\
Yichao Wu&100.00\%&51.42\% &68.35\%&76.34\%\\
Yacine Ait-Sahalia&91.06\%&51.69\% &68.20\%&76.43\%\\
Wenyang Zhang&91.17\%&51.69\% &68.20\%&76.43\%\\
Howell Tong&94.36\%&47.34\% &70.66\%&74.86\%\\
Chunming Zhang&89.83\%&52.03\% &68.00\%&76.55\%\\
Yingying Fan&82.52\%&44.17\% &72.39\%&73.69\%\\
Rui Song&91.21\%&52.65\% &67.64\%&76.77\%\\
Per Aslak Mykland&93.54\%&47.43\% &70.62\%&74.89\%\\
Bee Leng Lee&96.52\%&57.51\% &64.71\%&78.46\%\\
Runze Li&100.00\%&88.82\% &41.08\%&88.09\%\\
Jiancheng Jiang&75.59\%&29.41\% &79.72\%&67.83\%\\
Jiashun Jin&100.00\%&100.00\% &0.00\%&99.87\%\\
\hline
\end{tabular}%}
\end{table}

From Table \ref{Coauthorship}, we can find that Mixed-ISC, OCCAM and GeoNMF tend to classify authors in this table into the ``Carroll-Hall'' community (except Runze Li and Jiashun Jin for OCCAM method, which puts the two authors into the ``North Carolina'' community.), and this kind of classification is quite different from that of Mixed-SCORE. Particularly, the results show that Mixed-ISC returns extremely high weights on the ``Caroll-Hall'' community for the 20 authors. There are huge differences of the estimated PMFs between Mixed-ISC (and GeoNMF) and Mixed-SCORE on the following 9 authors:  J S Marron, Hao Helen Zhang, Yufeng Liu, Xiaotong Shen, Kung-Sik Chan, Howell Tong, Yingying Fan, Per Aslak Mykland, Jiancheng Jiang. We analyze Yingying Fan and Jiancheng Jiang in detail based on papers published by them in the top 4 journals during the time period of the Coauthorship network dataset.
\begin{itemize}
  \item For Yingying Fan, she published 6 papers on the top 4 journals while she coauthored with Jianqing Fan with 4 papers. Therefore, we tend to believe that Yingying Fan is more on the ``Carroll-Hall'' community since Jianqing Fan is more on this community.
  \item For Jiancheng Jiang, he published 10 papers on the top 4 journals while he coauthored with Jianqing Fan with 9 papers. Therefore, we tend to believe that Jiancheng Jiang is more on the ``Carroll-Hall'' community since Jianqing Fan is more on this community.
\end{itemize}
%Finally, combine the fact that the Coauthorship network is weak signal with the satisfactory performances of our Mixed-ISC on simulations and the three weak signal datasets (Simmons, Caltech and SNAP ego-networks), we tend to believe that our Mixed-ISC provide new insights on the estimation of authors' memberships for the Coauthorship network.
%%%%%
\subsection{Discussion on the choices of $\tau$}
There is no practical criterion for choosing a best Laplacian regularizer $\tau$ for Mixed-ISC at present. Since $\tau =c d$ where $c$ is a nonnegative constant and $d$ should have the same order as the observed average degree, there are two directions to study the effect of $\tau$ to Mixed-ISC: %by applying the Simmons, Caltech and SNAP ego-networks\footnote{For the discussion of the choice of $\tau$ for Mixed-ISC, for convenience, we do not apply Mixed-ISC on weak signal and strong signal networks for SNAP ego-networks separately. Instead, we apply Mixed-ISC on SNAP ego-networks by its three platforms directly.},
 one direction is changing $d$ when $c$ is fixed, and another direction is changing $c$ with fixed $d$.

For the first direction, we fixed $c=0.1$, and let $d$ have three choices $\bar{d}, d_{\mathrm{max}}$ and $\frac{d_{\mathrm{max}}+d_{\mathrm{min}}}{2}$, where $\bar{d}=\sum_{i=1}^{n}D(i,i)/n$. %Table \ref{real2errorsISCd} reports the error rates of Simmons and Caltech for different choice of $d$.
Table \ref{ErrorSNAPISCd} reports the Mean (SD) of mixed-Hamming error rates for different $d$. Numerical results in this two table show that Mixed-ISC works satisfactory for the $d$'s three different choices $\bar{d}, d_{\mathrm{max}}$ and $\frac{d_{\mathrm{max}}+d_{\mathrm{min}}}{2}$. Therefore, we set $\frac{d_{\mathrm{max}}+d_{\mathrm{min}}}{2}$ as the default $d$ in our Mixed-ISC, just as the default setting in the ISC method \cite{ISC}. % since it behaves slightly better than the other two choices on Simmons and Caltech while it performs similar on SNAP ego-networks as the case when $d$ is $\bar{d}$.
%\begin{table}[h!]
%%\footnotesize
%\centering
%\caption{Error rates on Simmons and Caltech with different choice of $d$ for Mixed-ISC when $\tau=0.1d$.}
%\label{real2errorsISCd}
%%\resizebox{\columnwidth}{!}{
%\begin{tabular}{cccccccccc}
%\toprule
%d &Simmons&Caltech\\
%\midrule
%$\bar{d}$&214/1137&95/590\\
%$d_{\mathrm{max}}$&129/1137&96/590\\
%$\frac{d_{\mathrm{max}}+d_{\mathrm{min}}}{2}$&128/1137&96/590\\
%\bottomrule
%\end{tabular}%}
%\end{table}

\begin{table}[h!]
\centering
\caption{Mean (SD) of mixed-Hamming error rates on SNAP ego-networks with different choice of $d$ for Mixed-ISC when $\tau=0.1d$.}
\label{ErrorSNAPISCd}
%\resizebox{\columnwidth}{!}{
\begin{tabular}{cccccccccc}
\toprule
d&Facebook&Google+&Twitter\\
\midrule
$\bar{d}$&0.2566(0.1293)&0.3151(0.1241)&0.2579(0.1366)\\
$d_{\mathrm{max}}$&0.2932(0.1549)&0.3221(0.1210)&0.2590(0.1379)\\
$\frac{d_{\mathrm{max}}+d_{\mathrm{min}}}{2}$&0.2582(0.1295)&0.3209(0.1280)&0.2582(0.1360)\\
\bottomrule
\end{tabular}%}
\end{table}
For the second direction, we fix $d$ as $\frac{d_{\mathrm{max}}+d_{\mathrm{min}}}{2}$, and change $c$. The results are displayed in Figure \ref{test_c}. We see that Mixed-ISC is robust on the choice of $c$ as long as $\tau$ is in the same order as the observed average degree. Combine with numerical results of Table \ref{ErrorSNAPISCd} and
Figure \ref{test_c}, we conclude that Mixed-ISC is insensitive to the choice of $\tau$.
\begin{figure}
    \centering
    \subfigure[]{
        \includegraphics[width=0.31\textwidth]{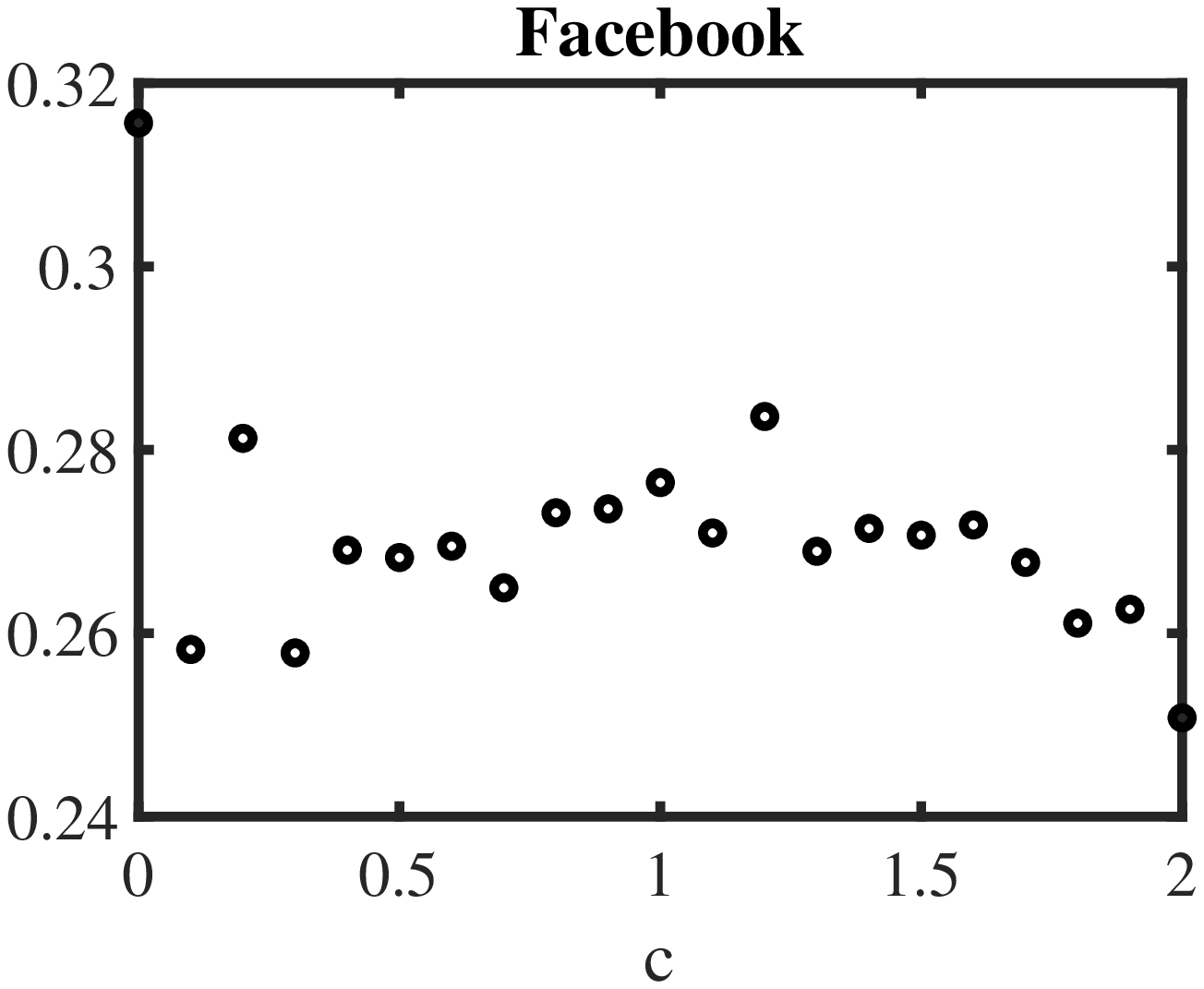}
    }
    \subfigure[]{
        \includegraphics[width=0.31\textwidth]{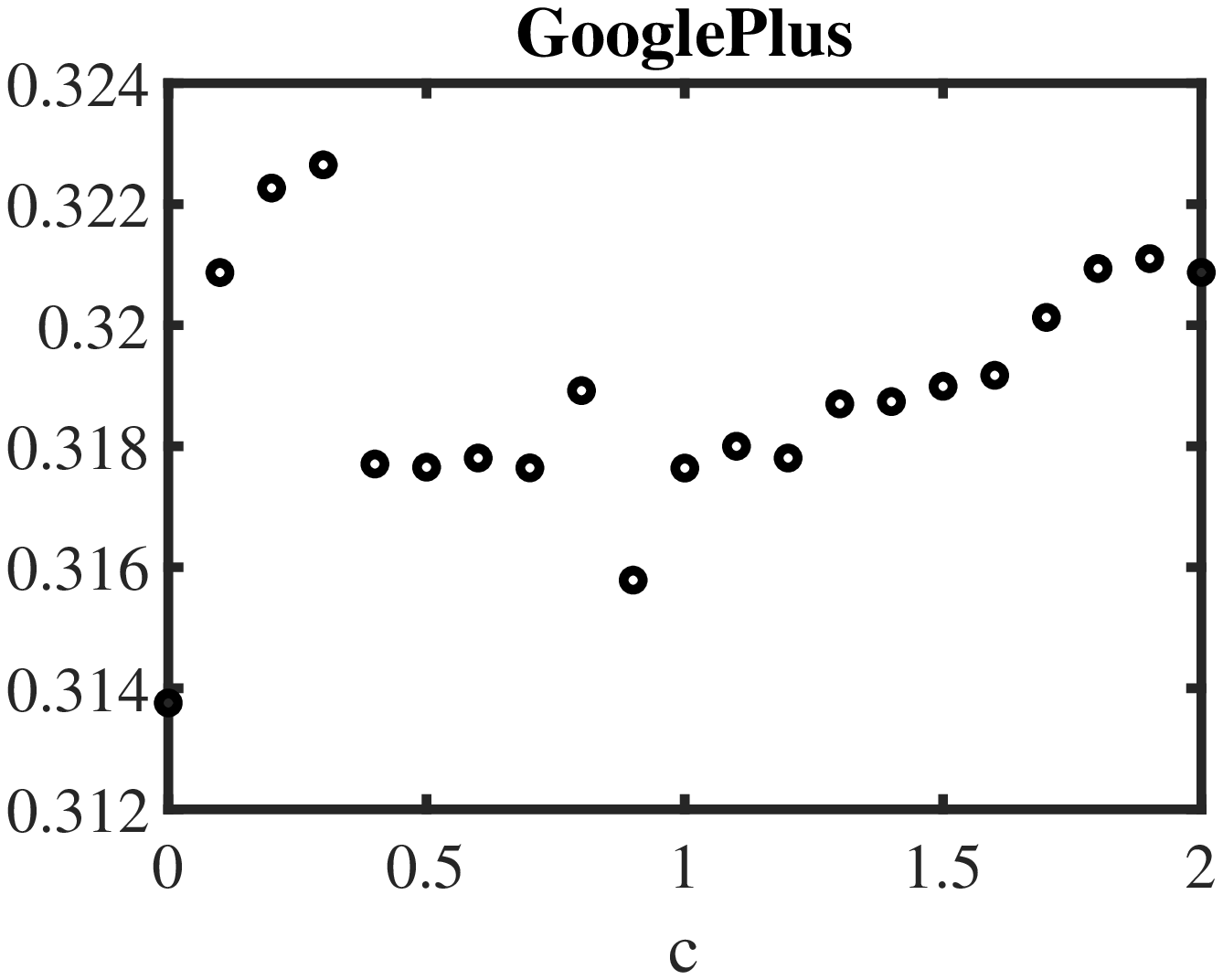}
    }
    \subfigure[]{
        \includegraphics[width=0.31\textwidth]{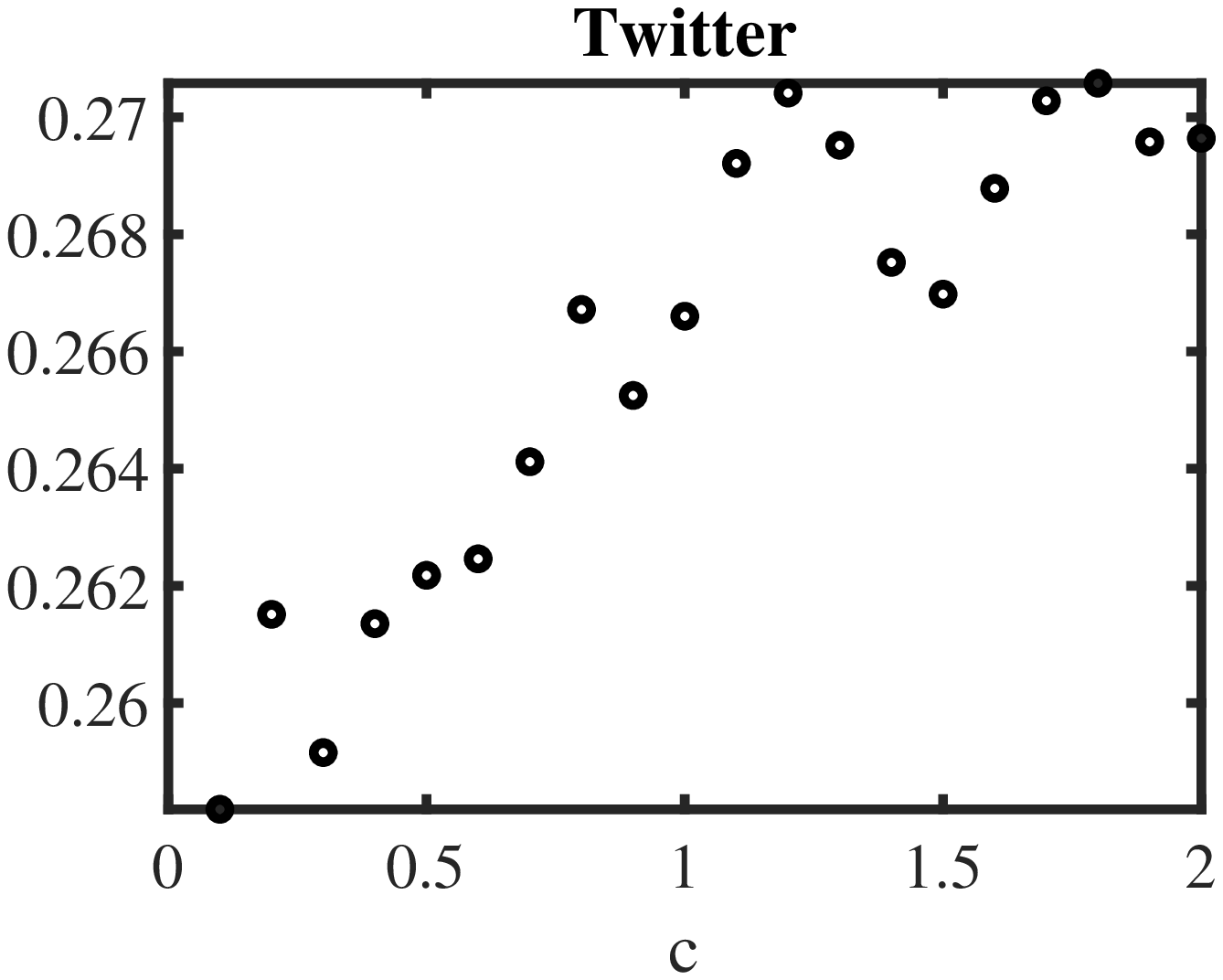}
    }
    \caption{Top three panels: mean of mixed-Hamming error rates on SNAP ego-networks for Mixed-ISC when $\tau=c\frac{d_{\mathrm{max}}+d_{\mathrm{min}}}{2}$, where $c$ is in $\{0,0.1,0.2,\ldots,2\}$. x-axis: $c$, y-axis: mean of mixed-Hamming error rates.}
    \label{test_c} %% label for entire figure
\end{figure}
%%%%%%
\section{Discussion}\label{sec8}
This paper makes one major contribution, Mixed-ISC method for mixed membership community detection under the DCMM model. Mixed-ISC can detect both strong signal and weak signal networks by considering the leading $K+1$ eigenvectors and eigenvalues of $L_{\tau}$. The Mixed-ISC method is an extension of the ISC method \cite{ISC}, and it inherits nice properties of ISC such that Mixed-ISC is insensitive to the choice of tuning parameters $\tau$. We showed the estimation consistency of Mixed-ISC under the DCMM model. Mixed-ISC enjoys satisfactory performances and it performs better than most of the benchmark methods both numerically and empirically. It remains unsolved that whether there exists optimal parameter $\tau$ both theoretically and numerically. In \cite{ali2018improved}, the authors studied the existence of
an optimal value $\alpha_{opt}$ of the parameter $\alpha$ for community detection methods based on $D^{-\alpha}AD^{-\alpha}$ for community detection problem. Recall that our Mixed-ISC is designed based on $D^{-1/2}_{\tau}AD^{-1/2}_{\tau}$, we argue that whether there exist optimal $\alpha_{0}$ and $\beta_{0}$ such that mixed membership community detection method designed based on $D^{\alpha_{0}}_{\tau}A^{\beta_{0}}D^{\alpha_{0}}_{\tau}$ outperforms methods designed based on $D^{\alpha}_{\tau}A^{\beta}D^{\alpha}_{\tau}$ for any choices of $\alpha$ and $\beta$. For reasons of space, we leave studies of these problems to the future.

\section*{Acknowledgements}
%The authors would like to thank the anonymous referees, an Associate Editor and the Editor for their constructive comments that improved the quality of this paper.
The authors would like to thank Dr. Yuan Zhang (the first author of the OCCAM method \cite{OCCAM}) for sharing the SNAP ego-networks with us.
\bibliographystyle{abbrv}%{DeGruyter}
\bibliography{reference}

\begin{thebibliography}{10}

\bibitem{MMSB}
E.~M. {Airoldi}, D.~M. {Blei}, S.~E. {Fienberg}, and E.~P. {Xing}.
\newblock Mixed membership stochastic blockmodels.
\newblock {\em Journal of Machine Learning Research}, 9:1981--2014, 2008.

\bibitem{ali2018improved}
H.~T. {Ali} and R.~{Couillet}.
\newblock Improved spectral community detection in large heterogeneous
  networks.
\newblock {\em Journal of Machine Learning Research}, 18(225):1--49, 2018.

\bibitem{anandkumar2014tensor}
A.~Anandkumar, R.~Ge, D.~Hsu, and S.~M. Kakade.
\newblock A tensor approach to learning mixed membership community models.
\newblock {\em The Journal of Machine Learning Research}, 15(1):2239--2312,
  2014.

\bibitem{Cai2016survey}
Q.~Cai, L.~Ma, M.~Gong, and D.~Tian.
\newblock A survey on network community detection based on evolutionary
  computation.
\newblock {\em International Journal of Bio-Inspired Computation},
  8(2):84–98, 2016.

\bibitem{CMM}
Y.~{Chen}, X.~{Li}, and J.~{Xu}.
\newblock Convexified modularity maximization for degree-corrected stochastic
  block models.
\newblock {\em Annals of Statistics}, 46(4):1573--1602, 2018.

\bibitem{Fortunato2016Community}
S.~Fortunato and D.~Hric.
\newblock Community detection in networks: A user guide.
\newblock {\em Physics Reports}, 659:1--44, 2016.

\bibitem{goldenberg2010a}
A.~{Goldenberg}, A.~X. {Zheng}, S.~E. {Fienberg}, and E.~M. {Airoldi}.
\newblock A survey of statistical network models.
\newblock {\em Foundations and Trends® in Machine Learning archive},
  2(2):129--233, 2010.

\bibitem{SBM}
P.~W. {Holland}, K.~B. {Laskey}, and S.~{Leinhardt}.
\newblock Stochastic blockmodels: First steps.
\newblock {\em Social Networks}, 5(2):109--137, 1983.

\bibitem{javed2018community}
M.~A. Javed, M.~S. Younis, S.~Latif, J.~Qadir, and A.~Baig.
\newblock Community detection in networks: A multidisciplinary review.
\newblock {\em Journal of Network and Computer Applications}, 108:87--111,
  2018.

\bibitem{ji2016coauthorship}
P.~{Ji} and J.~{Jin}.
\newblock Coauthorship and citation networks for statisticians.
\newblock {\em The Annals of Applied Statistics}, 10(4):1779--1812, 2016.

\bibitem{SCORE}
J.~{Jin}.
\newblock {Fast community detection by SCORE}.
\newblock {\em Annals of Statistics}, 43(1):57--89, 2015.

\bibitem{jin2017a}
J.~{Jin} and Z.~T. {Ke}.
\newblock A sharp lower bound for mixed-membership estimation.
\newblock {\em arXiv preprint arXiv:1709.05603}, 2017.

\bibitem{mixedSCORE}
J.~{Jin}, Z.~T. {Ke}, and S.~{Luo}.
\newblock Estimating network memberships by simplex vertex hunting.
\newblock {\em arXiv preprint arXiv:1708.07852}, 2017.

\bibitem{SCORE+}
J.~{Jin}, Z.~T. {Ke}, and S.~{Luo}.
\newblock Score+ for network community detection.
\newblock {\em arXiv preprint arXiv:1811.05927}, 2018.

\bibitem{SLIM}
B.~Jing, T.~Li, N.~Ying, and X.~Yu.
\newblock Community detection in sparse networks using the symmetrized
  laplacian inverse matrix (slim).
\newblock {\em Statistica Sinica}, 2021.

\bibitem{DCSBM}
B.~{Karrer} and M.~E.~J. {Newman}.
\newblock Stochastic blockmodels and community structure in networks.
\newblock {\em Physical Review E}, 83(1):16107, 2011.

\bibitem{karwa2016discussion}
V.~Karwa and S.~Petrovi{\'c}.
\newblock Discussion of" coauthorship and citation networks for statisticians".
\newblock {\em The Annals of Applied Statistics}, 10(4):1827--1834, 2016.

\bibitem{lei2015consistency}
J.~Lei and A.~Rinaldo.
\newblock Consistency of spectral clustering in stochastic block models.
\newblock {\em The Annals of Statistics}, 43(1):215--237, 2015.

\bibitem{leskovec2016snap}
J.~Leskovec and R.~Sosi{\v{c}}.
\newblock Snap: A general-purpose network analysis and graph-mining library.
\newblock {\em ACM Transactions on Intelligent Systems and Technology (TIST)},
  8(1):1--20, 2016.

\bibitem{li2020hierarchical}
T.~Li, L.~Lei, S.~Bhattacharyya, K.~Van~den Berge, P.~Sarkar, P.~J. Bickel, and
  E.~Levina.
\newblock Hierarchical community detection by recursive partitioning.
\newblock {\em Journal of the American Statistical Association}, pages 1--39,
  2020.

\bibitem{Tutorial}
U.~{Luxburg}.
\newblock A tutorial on spectral clustering.
\newblock {\em Statistics and Computing}, 17(4):395--416, 2007.

\bibitem{GeoNMF}
X.~{Mao}, P.~{Sarkar}, and D.~{Chakrabarti}.
\newblock On mixed memberships and symmetric nonnegative matrix factorizations.
\newblock In {\em International Conference on Machine Learning}, pages
  2324--2333, 2017.

\bibitem{SNAPego}
J.~{McAuley} and J.~{Leskovec}.
\newblock Learning to discover social circles in ego networks.
\newblock {\em In Advances in Neural Information Processing Systems},
  25:548--556, 2012.

\bibitem{mcauley2014discovering}
J.~Mcauley and J.~Leskovec.
\newblock Discovering social circles in ego networks.
\newblock {\em ACM Transactions on Knowledge Discovery from Data (TKDD)},
  8(1):1--28, 2014.

\bibitem{RSC}
T.~{Qin} and K.~{Rohe}.
\newblock {Regularized spectral clustering under the degree-corrected
  stochastic blockmodel}.
\newblock In {\em Advances in Neural Information Processing Systems 26}, pages
  3120--3128, 2013.

\bibitem{ISC}
H.~{Qing} and J.~{Wang}.
\newblock An improved spectral clustering method for community detection under
  the degree-corrected stochastic blockmodel.

\bibitem{traud2011comparing}
A.~L. {Traud}, E.~D. {Kelsic}, P.~J. {Mucha}, and M.~A. {Porter}.
\newblock {Comparing community structure to characteristics in online
  collegiate social network}.
\newblock {\em Siam Review}, 53(3):526--543, 2011.

\bibitem{weyl1912das}
H.~{Weyl}.
\newblock Das asymptotische verteilungsgesetz der eigenwerte linearer
  partieller differentialgleichungen (mit einer anwendung auf die theorie der
  hohlraumstrahlung).
\newblock {\em Mathematische Annalen}, 71(4):441--479, 1912.

\bibitem{yu2015a}
Y.~{Yu}, T.~{Wang}, and R.~J. {Samworth}.
\newblock {A useful variant of the Davis--Kahan theorem for statisticians}.
\newblock {\em Biometrika}, 102(2):315--323, 2015.

\bibitem{OCCAM}
Y.~{Zhang}, E.~{Levina}, and J.~{Zhu}.
\newblock {Detecting overlapping communities in networks using spectral
  methods}.
\newblock {\em SIAM Journal on Mathematics of Data Science}, 2(2):265--283,
  2020.

\end{thebibliography}
\appendix
\section{Error rates on Simmons and Caltech}
OCCAM, Mixed-SCORE and GeoNMF perform poor on the two weak signal networks Simmons and Caltech as shown in Table \ref{real2errors} while our Mixed-ISC significantly outperforms the three competitors on the two networks.
\begin{table}[t!]
\centering
\caption{Error rates on the two weak signal networks Simmons and Caltech.}
\label{real2errors}\par
%\resizebox{\linewidth}{!}{
\begin{tabular}{cccccc}
\hline
&Mixed-ISC&OCCAM&Mixed-SCORE&GeoNMF\\
\hline
Simmons&\textbf{128/1137}&268/1137&261/1137&263/1137\\
Caltech&\textbf{96/590}&192/590&174/590&208/590\\
\hline
\end{tabular}%}
\end{table}
\section{Proof for Mixed-ISC}
\subsection{Proof of Lemma \ref{boundL}.}
\begin{proof}
Since
\begin{align*}
\|L_{\tau}-\mathscr{L}_{\tau}\|&=\|D_{\tau}^{-1/2}AD_{\tau}^{-1/2}-\mathscr{D}^{-1/2}_{\tau}\Omega\mathscr{D}^{-1/2}_{\tau}\|\\
&=\|D_{\tau}^{-1/2}AD_{\tau}^{-1/2}-D_{\tau}^{-1/2}A\mathscr{D}_{\tau}^{-1/2}+D_{\tau}^{-1/2}A\mathscr{D}_{\tau}^{-1/2}-D_{\tau}^{-1/2}\Omega\mathscr{D}_{\tau}^{-1/2}\\
&~~~+D_{\tau}^{-1/2}\Omega\mathscr{D}_{\tau}^{-1/2}-\mathscr{D}^{-1/2}_{\tau}\Omega\mathscr{D}^{-1/2}_{\tau}\|\\
&\leq \|D_{\tau}^{-1/2}AD_{\tau}^{-1/2}-D_{\tau}^{-1/2}A\mathscr{D}_{\tau}^{-1/2}\|+\|D_{\tau}^{-1/2}A\mathscr{D}_{\tau}^{-1/2}-D_{\tau}^{-1/2}\Omega\mathscr{D}_{\tau}^{-1/2}\|\\
&~~~+\|D_{\tau}^{-1/2}\Omega\mathscr{D}_{\tau}^{-1/2}-\mathscr{D}^{-1/2}_{\tau}\Omega\mathscr{D}^{-1/2}_{\tau}\|\\
&=\|D_{\tau}^{-1/2}A(D_{\tau}^{-1/2}-\mathscr{D}_{\tau}^{-1/2})\|+\|D_{\tau}^{-1/2}(A-\Omega)\mathscr{D}_{\tau}^{-1/2}\|\\
&~~~+\|(D_{\tau}^{-1/2}-\mathscr{D}_{\tau}^{-1/2})\Omega\mathscr{D}_{\tau}^{-1/2}\|\\
&=\|D_{\tau}^{-1/2}D^{1/2}D^{-1/2}AD^{-1/2}D^{1/2}(D_{\tau}^{-1/2}-\mathscr{D}_{\tau}^{-1/2})\|+\|D_{\tau}^{-1/2}(A-\Omega)\mathscr{D}_{\tau}^{-1/2}\|\\
&~~~+\|(D_{\tau}^{-1/2}-\mathscr{D}_{\tau}^{-1/2})\mathscr{D}^{1/2}\mathscr{D}^{-1/2}\Omega\mathscr{D}^{-1/2}\mathscr{D}^{1/2}\mathscr{D}_{\tau}^{-1/2}\|\\
&\leq\|D_{\tau}^{-1/2}D^{1/2}\|\|D^{-1/2}AD^{-1/2}\|\|D^{1/2}\|\|D_{\tau}^{-1/2}-\mathscr{D}_{\tau}^{-1/2}\|+\|D_{\tau}^{-1/2}(A-\Omega)\mathscr{D}_{\tau}^{-1/2}\|\\
&~~~+\|D_{\tau}^{-1/2}-\mathscr{D}_{\tau}^{-1/2}\|\|\mathscr{D}^{1/2}\|\|\mathscr{D}^{-1/2}\Omega\mathscr{D}^{-1/2}\|\|\mathscr{D}^{1/2}\mathscr{D}_{\tau}^{-1/2}\|\\
&\leq\|D_{\tau}^{-1/2}D^{1/2}\|\|D^{1/2}\|\|D_{\tau}^{-1/2}-\mathscr{D}_{\tau}^{-1/2}\|+\|D_{\tau}^{-1/2}(A-\Omega)\mathscr{D}_{\tau}^{-1/2}\|\\
&~~~+\|D_{\tau}^{-1/2}-\mathscr{D}_{\tau}^{-1/2}\|\|\mathscr{D}^{1/2}\|\|\mathscr{D}^{1/2}\mathscr{D}_{\tau}^{-1/2}\|\\
&=\frac{\Delta_{\mathrm{max}}}{\sqrt{\tau+\Delta_{\mathrm{max}}}}\|D_{\tau}^{-1/2}-\mathscr{D}_{\tau}^{-1/2}\|+\|D_{\tau}^{-1/2}(A-\Omega)\mathscr{D}_{\tau}^{-1/2}\|\\
&~~~+\frac{\delta_{\mathrm{max}}}{\sqrt{\tau+\delta_{\mathrm{max}}}}\|D_{\tau}^{-1/2}-\mathscr{D}_{\tau}^{-1/2}\|\\
&=(\frac{\Delta_{\mathrm{max}}}{\sqrt{\tau+\Delta_{\mathrm{max}}}}+\frac{\delta_{\mathrm{max}}}{\sqrt{\tau+\delta_{\mathrm{max}}}})\|D_{\tau}^{-1/2}-\mathscr{D}_{\tau}^{-1/2}\|+\|D_{\tau}^{-1/2}(A-\Omega)\mathscr{D}_{\tau}^{-1/2}\|\\
&\leq (\frac{\Delta_{\mathrm{max}}}{\sqrt{\tau+\Delta_{\mathrm{max}}}}+\frac{\delta_{\mathrm{max}}}{\sqrt{\tau+\delta_{\mathrm{max}}}})\|D_{\tau}^{-1/2}-\mathscr{D}_{\tau}^{-1/2}\|+\|D_{\tau}^{-1/2}\|\|A-\Omega\|\|\mathscr{D}_{\tau}^{-1/2}\|\\
&\leq (\frac{\Delta_{\mathrm{max}}}{\sqrt{\tau+\Delta_{\mathrm{max}}}}+\frac{\delta_{\mathrm{max}}}{\sqrt{\tau+\delta_{\mathrm{max}}}})\|D_{\tau}^{-1/2}-\mathscr{D}_{\tau}^{-1/2}\|+\frac{1}{\sqrt{(\tau+\Delta_{\mathrm{min}})(\tau+\delta_{\mathrm{min}})}}\|A-\Omega\|.
\end{align*}
Therefore, to bound $\|L_{\tau}-\mathscr{L}_{\tau}\|$, now we need to bound $\|D^{-1/2}_{\tau}-\mathscr{D}^{-1/2}_{\tau}\|$ and $\|A-\Omega\|$.

When assumptions (\ref{a1}) and (\ref{a2}) hold, by Lemma 3.2 of \cite{mixedSCORE}, we have: with probability $1-o(n^{-3})$,
\begin{align}\label{diffAOmega}
\|A-\Omega\|\leq C\sqrt{\mathrm{log}(n)\theta_{\mathrm{max}}\|\theta\|_{1}},
\end{align}
 where $C$ is a positive constant number.

 Next we bound $\|D_{\tau}^{-1/2}-\mathscr{D}_{\tau}^{-1/2}\|$,
\begin{align*}
\|D_{\tau}^{-1/2}-\mathscr{D}_{\tau}^{-1/2}\|&=\|(\mathscr{D}_{\tau}^{-1/2}D^{1/2}_{\tau}-I)D_{\tau}^{-1/2}\|\\
&\leq \|\mathscr{D}_{\tau}^{-1/2}D^{1/2}_{\tau}-I\|\|D_{\tau}^{-1/2}\|\\
&=\frac{\|I-\mathscr{D}_{\tau}^{-1/2}D^{1/2}_{\tau}\|}{\sqrt{\tau+\Delta_{\mathrm{min}}}}.
\end{align*}
We obtain the bound of $\|I-\mathscr{D}_{\tau}^{-1/2}D^{1/2}_{\tau}\|$ by below procedure: since $\|I-\mathscr{D}_{\tau}^{-1/2}D^{1/2}_{\tau}\|=\underset{1\leq i\leq n}{\mathrm{max}}|1-\sqrt{\frac{D_{\tau}(i,i)}{\mathscr{D}_{\tau}(i,i)}}|$, set $i^{*}=\mathrm{arg~max}_{1\leq i\leq n}|1-\sqrt{\frac{D_{\tau}(i,i)}{\mathscr{D}_{\tau}(i,i)}}|$. If $1\geq \frac{D_{\tau}(i^{*},i^{*})}{\mathscr{D}_{\tau}(i^{*},i^{*})}$, we have
\begin{align*}
\|I-\mathscr{D}_{\tau}^{-1/2}D^{1/2}_{\tau}\|=1-\sqrt{\frac{D_{\tau}(i^{*},i^{*})}{\mathscr{D}_{\tau}(i^{*},i^{*})}}\leq 1-\sqrt{\frac{\tau+\Delta_{\mathrm{min}}}{\tau+\delta_{\mathrm{max}}}}.
\end{align*}
 If $1\leq \frac{D_{\tau}(i^{*},i^{*})}{\mathscr{D}_{\tau}(i^{*},i^{*})}$, we have
\begin{align*}
\|I-\mathscr{D}_{\tau}^{-1/2}D^{1/2}_{\tau}\|=\sqrt{\frac{D_{\tau}(i^{*},i^{*})}{\mathscr{D}_{\tau}(i^{*},i^{*})}}-1\leq \sqrt{\frac{\tau+\Delta_{\mathrm{max}}}{\tau+\delta_{\mathrm{min}}}}-1.
\end{align*}
Therefore, we have
\begin{align*}
\|I-\mathscr{D}_{\tau}^{-1/2}D^{1/2}_{\tau}\|\leq \mathrm{max}\{1-\sqrt{\frac{\tau+\Delta_{\mathrm{min}}}{\tau+\delta_{\mathrm{max}}}}, \sqrt{\frac{\tau+\Delta_{\mathrm{max}}}{\tau+\delta_{\mathrm{min}}}}-1\},
\end{align*}
which gives
\begin{align}\label{diffD}
\|D_{\tau}^{-1/2}-\mathscr{D}_{\tau}^{-1/2}\|\leq\frac{\mathrm{max}\{1-\sqrt{\frac{\tau+\Delta_{\mathrm{min}}}{\tau+\delta_{\mathrm{max}}}}, \sqrt{\frac{\tau+\Delta_{\mathrm{max}}}{\tau+\delta_{\mathrm{min}}}}-1\}}{\sqrt{\tau+\Delta_{\mathrm{min}}}}.
\end{align}
Combine (\ref{diffAOmega}) and (\ref{diffD}), we obtain the bound of $\|L_{\tau}-\mathscr{L}_{\tau}\|$ as below
\begin{align*}
\|L_{\tau}-\mathscr{L}_{\tau}\|&\leq (\frac{\Delta_{\mathrm{max}}}{\sqrt{\tau+\Delta_{\mathrm{max}}}}+\frac{\delta_{\mathrm{max}}}{\sqrt{\tau+\delta_{\mathrm{max}}}})\|D_{\tau}^{-1/2}-\mathscr{D}_{\tau}^{-1/2}\|+\frac{1}{\sqrt{(\tau+\Delta_{\mathrm{min}})(\tau+\delta_{\mathrm{min}})}}\|A-\Omega\|\\
&\leq (\frac{\Delta_{\mathrm{max}}}{\sqrt{\tau+\Delta_{\mathrm{max}}}}+\frac{\delta_{\mathrm{max}}}{\sqrt{\tau+\delta_{\mathrm{max}}}})\frac{\mathrm{max}\{1-\sqrt{\frac{\tau+\Delta_{\mathrm{min}}}{\tau+\delta_{\mathrm{max}}}}, \sqrt{\frac{\tau+\Delta_{\mathrm{max}}}{\tau+\delta_{\mathrm{min}}}}-1\}}{\sqrt{\tau+\Delta_{\mathrm{min}}}}\\
&~~~+\frac{C\sqrt{\mathrm{log}(n)\theta_{\mathrm{max}}\|\theta\|_{1}}}{\sqrt{(\tau+\Delta_{\mathrm{min}})(\tau+\delta_{\mathrm{min}})}}.
\end{align*}
\end{proof}
\subsection{Proof of Lemma \ref{boundXstar}.}
\begin{proof}
For notation convenience, we define eight matrices as below
	\begin{align*}
	\hat{V}_{K+1}&=[\hat{\eta}_{1}, \hat{\eta}_{2}, \ldots, \hat{\eta}_{K}, \hat{\eta}_{K+1}], ~~\hat{V}_{K}=[\hat{\eta}_{1}, \hat{\eta}_{2}, \ldots, \hat{\eta}_{K}],\\
	V_{K+1}&=[\eta_{1}, \eta_{2}, \ldots, \eta_{K}, \eta_{K+1}], ~~V_{K}=[\eta_{1}, \eta_{2}, \ldots, \eta_{K}],\\
	\hat{E}_{K+1}&=\mathrm{diag}(\hat{\lambda}_{1}, \hat{\lambda}_{2}, \ldots, \hat{\lambda}_{K}, \hat{\lambda}_{K+1}), ~~\hat{E}_{K}=\mathrm{diag}(\hat{\lambda}_{1}, \hat{\lambda}_{2}, \ldots, \hat{\lambda}_{K}),\\
	E_{K+1}&=\mathrm{diag}(\lambda_{1}, \lambda_{2}, \ldots, \lambda_{K}, 0), ~~E_{K}=\mathrm{diag}(\lambda_{1}, \lambda_{2}, \ldots, \lambda_{K}).
	\end{align*}
	Then we have
	\begin{align*}
	\hat{X}=\hat{V}_{K+1}\hat{E}_{K+1},~~X=V_{K+1}E_{K+1},
	\end{align*}
	which gives that
	\begin{align}
	\|\hat{X}-X\|_{F}&=\|\hat{V}_{K+1}\hat{E}_{K+1}-V_{K+1}E_{K+1}\|_{F}\notag\\
	&=\|\hat{V}_{K+1}\hat{E}_{K+1}-\hat{V}_{K+1}E_{K+1}+\hat{V}_{K+1}E_{K+1}-V_{K+1}E_{K+1}\|_{F}\notag\\
	&\leq \|\hat{V}_{K+1}\hat{E}_{K+1}-\hat{V}_{K+1}E_{K+1}\|_{F}+\|\hat{V}_{K+1}E_{K+1}-V_{K+1}E_{K+1}\|_{F}\notag\\
	&=\|\hat{V}_{K+1}(\hat{E}_{K+1}-E_{K+1})\|_{F}+\|(\hat{V}_{K+1}-V_{K+1})E_{K+1}\|_{F}\notag\\
	&=\|\hat{E}_{K+1}-E_{K+1}\|_{F}+\|(\hat{V}_{K+1}-V_{K+1})E_{K+1}\|_{F}\label{A3}\\
	&=\|\hat{E}_{K+1}-E_{K+1}\|_{F}+\|(\hat{V}_{K}-V_{K})E_{K}\|_{F}\notag\\
	&\leq\|\hat{E}_{K+1}-E_{K+1}\|_{F}+\|\hat{V}_{K}-V_{K}\|_{F}\|E_{K}\|_{F}\notag\\
	&=\sqrt{\|\hat{E}_{K}-E_{K}\|^{2}_{F}+\hat{\lambda}^{2}_{K+1}}+\|\hat{V}_{K}-V_{K}\|_{F}\|E_{K}\|_{F}\label{A4}.
	\end{align}
	Note that (\ref{A3}) holds because we only consider eigenvectors with unit-norms. Next we bound the four terms: $\|\hat{E}_{K}-E_{K}\|^{2}_{F}+\hat{\lambda}^{2}_{K+1}, \|\hat{V}_{K}-V_{K}\|_{F}$ and $\|E_{K}\|_{F}$.

By Weyl's inequality \cite{weyl1912das}, with probability at least $1-o(n^{-3})$, we have
\begin{align}\label{A5}
\|\hat{E}_{K}-E_{K}\|^{2}_{F}+\hat{\lambda}^{2}_{K+1}\leq (K+1)\|L_{\tau}-\mathscr{L}_{\tau}\|^{2}\leq (K+1)err^{2}_{n}.
\end{align}
We obtain the bound of $\|\hat{V}_{K}-V_{K}\|$ by combining a variant of the Davis-Kahan theorem \cite{yu2015a} and Lemma \ref{boundL}. The variant of the Davis-Kahan theorem is given below
\begin{thm}\label{VariantDK}
Let $\Sigma, \hat{\Sigma}\in \mathcal{R}^{p\times p}$ be symmetric with eigenvalues $\lambda_{1}\geq \ldots\geq \lambda_{p}$ and $\hat{\lambda}_{1}\geq\ldots\geq \hat{\lambda}_{p}$. Fix $1\leq r\leq s\leq p$ and assum that $\mathrm{min}(\lambda_{r-1}-\lambda_{r},\lambda_{s}-\lambda_{s+1})>0$, where we define $\lambda_{0}=\infty$ and $\lambda_{p+1}=\infty$. Let $d=s-r+1$, and let $V=(v_{r}, v_{r+1},\ldots, v_{s})\in \mathcal{R}^{p\times d}$ and $\hat{V}=(\hat{v}_{r}, \hat{v}_{r+1},\ldots, \hat{v}_{s})\in \mathcal{R}^{p\times d}$ have orthonormal columns satisfying $\Sigma v_{j}=\lambda_{j}v_{j}$ and $\hat{\Sigma}\hat{v}_{j}=\hat{\lambda}_{j}\hat{v}_{j}$ for $j=r,r+1, \ldots, s$. Then there exists an orthogonal matrix $\mathcal{O}\in \mathcal{R}^{d\times d}$ such that
\begin{align*}
\|\hat{V}-V\mathcal{O}\|_{F}\leq \frac{2^{3/2}\mathrm{min}(d^{1/2}\|\hat{\Sigma}-\Sigma\|, \|\hat{\Sigma}-\Sigma\|_{F} )}{\mathrm{min}(\lambda_{r-1}-\lambda_{r}, \lambda_{s}-\lambda_{s+1})}.
\end{align*}
\end{thm}

Set $\hat{\Sigma}=L_{\tau}, \Sigma=\mathscr{L}_{\tau}, r=1, s=K$, then $d=K, \hat{V}=\hat{V}_{K}, V=V_{K}$. By assumption (\ref{a3}), combine with Lemma \ref{boundL}, we have
\begin{align}\label{A6}
\|\hat{V}_{K}-V_{K}\|_{F}\leq\frac{\sqrt{8K}\|L_{\tau}-\mathscr{L}_{\tau}\|}{\lambda_{K}}\leq \frac{\sqrt{8K}err_{n}}{\lambda_{K}},
\end{align}
where we do not consider the orthogonal matrix since $\hat{V}_{K}$ and $V_{K}$ are arranged by the leading $K$ unit-norm eigenvectors of $L_{\tau}$ and $\mathscr{L}_{\tau}$, respectively.

We bound $\|E_{K}\|_{F}$ by
\begin{align*}
\|E_{K}\|_{F}\leq \sqrt{K}\lambda_{1}.
\end{align*}
Meanwhile, since $\|\mathscr{D}^{-1/2}\Omega\mathscr{D}^{-1/2}\|=1$ and  $\lambda_{1}=\|\mathscr{L}_{\tau}\|$, we have
\begin{align*}
\lambda_{1}=\|\mathscr{L}_{\tau}\|&=\|\mathscr{D}_{\tau}^{-1/2}\Omega\mathscr{D}^{-1/2}_{\tau}\|\\
&=\|\mathscr{D}^{-1/2}_{\tau}\mathscr{D}^{1/2}\mathscr{D}^{-1/2}\Omega\mathscr{D}^{-1/2}\mathscr{D}^{1/2}\mathscr{D}^{-1/2}_{\tau}\|\\
&\leq\|\mathscr{D}^{-1/2}_{\tau}\mathscr{D}^{1/2}\|\|\mathscr{D}^{-1/2}\Omega\mathscr{D}^{-1/2}\|\|\mathscr{D}^{1/2}\mathscr{D}^{-1/2}_{\tau}\|\\
&=\|\mathscr{D}^{-1/2}_{\tau}\mathscr{D}^{1/2}\|^{2}\\
&=\frac{\delta_{\mathrm{max}}}{\tau+\delta_{\mathrm{max}}},
\end{align*}
which gives that
\begin{align}\label{A7}
\|E_{K}\|_{F}\leq \frac{\delta_{\mathrm{max}}\sqrt{K}}{\tau+\delta_{\mathrm{max}}}.
\end{align}
Combine (\ref{A5})-(\ref{A7}), we have
\begin{align*}
\|\hat{X}-X\|_{F}\leq (\sqrt{K+1}+\frac{K\delta_{\mathrm{max}}}{\lambda_{K}(\tau+\delta_{\mathrm{max}})})err_{n}.
\end{align*}
Since for any $i$, we have
\begin{align*}
\|\hat{X}^{*}_{i}-X^{*}_{i}\|_{2}\leq \frac{\|\hat{X}_{i}-X_{i}\|}{\mathrm{min}\{\|\hat{X}_{i}\|_{2},\|X_{i}\|_{2}\}},
\end{align*}
we have
\begin{align*}
\|\hat{X}^{*}-X^{*}\|_{F}\leq \frac{\|\hat{X}-X\|_{F}}{m}\leq \frac{1}{m}(\sqrt{K+1}+\frac{K\delta_{\mathrm{max}}}{\lambda_{K}(\tau+\delta_{\mathrm{max}})})err_{n}.
\end{align*}
\end{proof}
\subsection{Proof of Lemma \ref{boundV}.}
\begin{proof}
To prove this lemma, we follow similar procedure as Lemma SM1.5. in \cite{OCCAM}. Because all rows of $\hat{X}^{*}$ and $X^{*}$ have unit norms, the sample space of $\mathcal{F}$ is uniformly bounded in the unit ball in $\mathcal{R}^{K+1}$. Then we show that all cluster vertexes estimated by K-medians fall in the $l_{2}$ ball centered at origin with radius 3, which we denote as $\mathscr{R}$. Otherwise, if there exists an estimated cluster center $v$ outside $\mathscr{R}$, it is at least distance 2 away from any point assigned to its cluster.  Therefore, moving $v$ to any point inside the unit ball yields an improvement in the loss function since any two points inside the unit ball are at most distance 2 away from each other.

We first show the uniform convergence of $\mathcal{L}_{n}(\hat{X}^{*};U)$ to $\mathcal{L}(\mathcal{F};U)$ and then show the optimum of $\mathcal{L}_{n}(\hat{X}^{*};U)$ is close to that of $\mathcal{L}(\mathcal{F};U)$. Let $\hat{V}:=\mathrm{arg~min~}\mathcal{L}_{n}(\hat{X}^{*};U)$. We start with showing that
\begin{align}\label{AE1}
\underset{U\subset\mathscr{R}}{\mathrm{sup}}|\mathcal{L}_{n}(\hat{X};U)-\mathcal{L}_{n}(X^{*};U)|\leq\frac{\|\hat{X}^{*}-X^{*}\|_{F}}{\sqrt{n}}.
\end{align}
To prove (\ref{AE1}), for each $i$, let $\hat{u}$ and $u$ be (possibly identical) rows in $U$ that are close to $\hat{X}^{*}_{i}$ and $X^{*}_{i}$ respectively in $l_{2}$ norm. We have
\begin{align*}
\|X^{*}_{i}-u\|_{F}-\|\hat{X}^{*}_{i}-\hat{u}\|_{F}\leq \|X^{*}_{i}-\hat{X}^{*}_{i}\|_{F}
\end{align*}
and similarly, $\|\hat{X}^{*}_{i}-\hat{u}\|_{F}-\|X^{*}_{i}-u\|_{F}\leq \|X^{*}_{i}-\hat{X}^{*}_{i}\|_{F}$. Hence $\|\hat{X}^{*}_{i}-\hat{u}\|_{F}-\|X^{*}_{i}-u\|_{F}\leq \|X^{*}_{i}-\hat{X}^{*}_{i}\|_{F}$. Combining this inequalities for all rows, we have
\begin{align}
&|\mathcal{L}_{n}(\hat{X}^{*};U)-\mathcal{L}_{n}(X^{*};U)|=|\frac{1}{n}\sum_{i=1}^{n}(\|\hat{X}^{*}_{i}-\hat{u}\|_{F}-\|X^{*}_{i}-u\|_{F})|\notag\\
&\leq \sqrt{\frac{1}{n}\sum_{i=1}^{n}\|\hat{X}^{*}_{i}-X^{*}_{i}\|^{2}_{F}}=\frac{1}{\sqrt{n}}\|\hat{X}^{*}-X^{*}\|_{F}.\label{AE2}
\end{align}
Then since that (\ref{AE2})holds for any $U$, the uniform bound (\ref{AE1}) follows.

For simplicity, we introduce the notation ``$U\subset \mathscr{R}$'', by which we mean that the rows of a matrix $U$ belong to the set $\mathscr{R}$. We now derive the bound for $\mathrm{sup}_{U\subset\mathscr{R}}|\mathcal{L}_{n}(X^{*};U)-\mathcal{L}(\mathcal{F};U)|$, which, without taking the supremum, is easily bounded by Bernstein's inequality. To tackle the uniform bound, we employ an $\epsilon$-net. there exists an $\epsilon$-net $\mathscr{R}_{\epsilon}$, with size $|\mathscr{R}_{\epsilon}|\leq C_{\mathscr{R}}\frac{(K+1)}{\epsilon}\mathrm{log}(\frac{(K+1)}{\epsilon})$, where $C_{\mathscr{R}}$ is a constant. For any $\tilde{U}\in \mathscr{R}_{\epsilon}, \tilde{U}\in \mathcal{R}^{K\times (K+1)}$, notice that $\mathrm{min}_{1\leq k\leq K}\|X^{*}_{i}-\tilde{U}_{k}\|_{F}$ is a random variable uniformly bounded by 6 with expectation $\mathcal{L}(\mathcal{F};\tilde{U})$ for each $i$. therefore, by Bernstein's inequality, for any $\varrho>0$, we have
\begin{align}\label{AE3}
\mathrm{P}(|\mathcal{L}_{n}(X^{*};\tilde{U})-\mathcal{L}(\mathcal{F};\tilde{U})|>\varrho)\leq \mathrm{exp}(-\frac{n\varrho^{2}}{8R^{2}_{m}+4R_{m}\varrho/3})=\mathrm{exp}(-\frac{n\varrho^{2}}{72+4\varrho}),
\end{align}
where $R_{m}$ is a universal upper bound on the radius of the sample space of latent node positions (rows of all true and estimated $X$).

The number of all such $\tilde{U}\subset\mathscr{R}_{\epsilon}$ is bounded by
\begin{align*}
|\{\tilde{U}: \tilde{U}\subset\mathscr{R}_{\epsilon}\}|=\bigl(\begin{smallmatrix}
    C_{\mathscr{R}}\frac{(K+1)}{\epsilon}\mathrm{log}(\frac{(K+1)}{\epsilon})\\
    K+1
  \end{smallmatrix}\bigr)\leq (C_{\mathscr{R}}\frac{(K+1)}{\epsilon}\mathrm{log}(\frac{(K+1)}{\epsilon}))^{K+1}.
\end{align*}
By the uniform bound, we have
\begin{align}\label{AE4}
  \mathrm{P}(\underset{\tilde{U}\subset\mathscr{R}_{\epsilon}}{\mathrm{sup}}|\mathcal{L}(X^{*};\tilde{U})-\mathcal{L}(\mathcal{F};\tilde{U})|>\varrho)<(C_{\mathscr{R}}\frac{(K+1)}{\epsilon}\mathrm{log}(\frac{(K+1)}{\epsilon}))^{K+1}\mathrm{exp}(-\frac{n\varrho^{2}}{72+4\varrho}).
\end{align}
The above shows the uniform convergence of the loss functions for $\tilde{U}$ from the $\epsilon$-net $\mathscr{R}_{\epsilon}$. We then expand it to the uniform convergence of all $U\subset \mathscr{R}$. For any $U\subset\mathscr{R}$, there exists $\tilde{U}\subset\mathscr{R}_{\epsilon}$ such that both $\mathcal{L}_{n}(\cdot; U)$ and $\mathcal{L}(\cdot; U)$ can be well approximated by $\mathcal{L}_{n}(\cdot; \tilde{U})$ and $\mathcal{L}(\cdot; \tilde{U})$, respectively. To emphasize the dependence of $\tilde{U}$ on $U$, we write $\tilde{U}=\tilde{U}(U)$. Formally, we now prove the following
\begin{align}
&\underset{U\subset\mathscr{R}}{\mathrm{sup}}|\mathcal{L}_{n}(X^{*};U)-\mathcal{L}_{n}(X^{*};\tilde{U}(U))|\leq \epsilon,\label{AE5}\\
&\underset{U\subset\mathscr{R}}{\mathrm{sup}}|\mathcal{L}(\mathcal{F};U)-\mathcal{L}(\mathcal{F};\tilde{U}(U))|\leq\epsilon\label{AE6}.
\end{align}
To prove (\ref{AE5}) and (\ref{AE6}), for any $U\subset\mathscr{R}$, let $\tilde{U}\subset\mathscr{R}_{\epsilon}$ be a matrix formed by concatenating the points in $\mathscr{R}_{\epsilon}$ that best approximate the rows in $U$. Notice that $\tilde{U}$ formed such way may contain less than $K$ rows, in this case, we arbitrarily pick points in $\mathscr{R}_{\epsilon}$ to enlarge $\tilde{U}$ to $K$ rows. For any $x\in \mathcal{R}^{K+1}$, let $u_{0}$ be the best approximation to $x$ among the rows of $U$ and $\tilde{u}_{0}$ be the best approximation to $u_{0}$ among the rows of $\tilde{U}$; let $\tilde{u}_{1}$ be the best approximation to $x$ among the rows of $\tilde{U}$ and let $u_{1}$ be the point among the rows of $U$ that is best approximated by $\tilde{u}_{1}$. Since $\|x-u_{0}\|_{2}\leq \|x-u_{1}\|_{2}\leq \|x-\tilde{u}_{1}\|_{2}+\|u_{1}-\tilde{u}_{1}\|_{2}\leq \|x-\tilde{u}_{1}\|_{2}+\epsilon$, and similarly, $\|x-\tilde{u}_{1}\|_{2}\leq \|x-\tilde{u}_{0}\|_{2}\leq \|x-u_{0}\|_{2}+\epsilon$, we have
\begin{align}\label{AE7}
|\mathrm{min}_{1\leq k\leq K}\|x-U_{k}\|_{2}-\mathrm{min}_{1\leq k\leq K}\|x-\tilde{U}_{k}\|_{2}|=|\|x-u_{0}\|_{2}-\|x-\tilde{u}_{1}\|_{2}|\leq \epsilon
\end{align}
which implies (\ref{AE5}) and (\ref{AE6}).

Combining (\ref{AE1}), (\ref{AE4}), (\ref{AE5}) and (\ref{AE6}), we have shown that with probability $\mathrm{P}_{2}(n,\epsilon,\varrho)=1-(C_{\mathscr{R}}\frac{(K+1)}{\epsilon}\mathrm{log}(\frac{(K+1)}{\epsilon}))\mathrm{exp}(-\frac{n\varrho^{2}}{72+4\varrho})$,
\begin{align}
&\underset{U\subset\mathscr{R}}{\mathrm{sup}}|\mathcal{L}(\hat{X}^{*};U)-\mathcal{L}(\mathcal{F};U)|\notag\\
&\leq \underset{U\subset\mathscr{R}}{\mathrm{sup}}|\mathcal{L}(\hat{X}^{*};U)-\mathcal{L}(X^{*};U)|+\underset{U\subset\mathscr{R}}{\mathrm{sup}}|\mathcal{L}(X^{*};U)-\mathcal{L}(X^{*};\tilde{U}(U))|\notag\\
&~~~+\underset{\tilde{U}\subset\mathscr{R}}{\mathrm{sup}}|\mathcal{L}(X^{*};\tilde{U})-\mathcal{L}(\mathcal{F};\tilde{U})|+\underset{V\subset\mathscr{R}}{\mathrm{sup}}|\mathcal{L}(\mathcal{F};U)-\mathcal{L}(\mathcal{F};\tilde{U}(U))|\notag\\
&\leq \frac{\|\hat{X}^{*}-X^{*}\|_{F}}{\sqrt{n}}+2\epsilon+\varrho. \label{AE8}
\end{align}
Finally, we use (\ref{AE8}) to bound $\|\hat{V}- V_{\mathcal{F}}\|_{F}$. Note that
\begin{align*}
&\mathcal{L}(\mathcal{F}; \hat{V})-\mathcal{L}(\mathcal{F}; V_{\mathcal{F}})\leq |\mathcal{L}(\mathcal{F}; \hat{V})-\mathcal{L}(\hat{X}^{*}; \hat{V})|\\
&~~~+(\mathcal{L}(\hat{X}^{*}; \hat{V})-\mathcal{L}(\hat{X}^{*}; V_{\mathcal{F}}))+|\mathcal{L}(\hat{X}^{*}; V_{\mathcal{F}})-\mathcal{L}(\mathcal{F}; V_{\mathcal{F}})|\\
&\leq 2\underset{U\subset\mathscr{R}}{\mathrm{sup}}|\mathcal{L}(\hat{X}^{*};U)-\mathcal{L}(\mathcal{F};U)|.
\end{align*}
Taking $\varrho=\epsilon=\frac{1}{m\sqrt{n}}(\sqrt{K+1}+\frac{K\delta_{\mathrm{max}}}{\lambda_{K}(\tau+\delta_{\mathrm{max}})})err_{n}$, define $\mathrm{P}_{2}(n,K)=\mathrm{P}_{2}(n,\epsilon,\varrho)$ with plug-in values of $\varrho$ and $\epsilon$. To sumarize, we have that with probability at least $1-o(n^{-3})-\mathrm{P}_{2}(n,K)$, the following holds:
\begin{align*}
\|\hat{V}- V_{\mathcal{F}}\|_{F}&\leq (\kappa K^{-1})^{-1}(\mathcal{L}(\mathcal{F};\hat{V})-\mathcal{L}(\mathcal{F}; V_{\mathcal{F}}))\\
&\leq\frac{2K}{\kappa}\underset{U\subset\mathscr{R}}{\mathrm{sup}}|\mathcal{L}(\hat{X}^{*};U)-\mathcal{L}(\mathcal{F};U)|\\
&\leq \frac{2K}{\kappa}(\frac{\|\hat{X}^{*}-X^{*}\|_{F}}{\sqrt{n}}+2\epsilon+\varrho)\\
&\leq \frac{CK}{m\sqrt{n}}(\sqrt{K+1}+\frac{K\delta_{\mathrm{max}}}{\lambda_{K}(\tau+\delta_{\mathrm{max}})})err_{n},
\end{align*}
where $C$ is a constant. And we have $\mathrm{P}(n,K)=1-o(n^{-3})-\mathrm{P}_{2}(n,K)$.
\end{proof}
\subsection{Proof of Lemma \ref{boundY}.}
\begin{proof}
Since
\begin{align*}
\|\hat{Y}-Y\|_{F}&=\|\mathrm{max}(\hat{X}^{*}\hat{V}'(\hat{V}\hat{V}')^{-1},0)-\mathrm{max}(X^{*}V'(VV')^{-1},0)\|_{F}\\
&\leq\|\hat{X}^{*}\hat{V}'(\hat{V}\hat{V}')^{-1}-X^{*}V'(VV')^{-1}\|_{F}\\
&=\|\hat{X}^{*}\hat{V}'(\hat{V}\hat{V}')^{-1}-\hat{X}^{*}V'(VV')^{-1}+\hat{X}^{*}V'(VV')^{-1}-X^{*}V'(VV')^{-1}\|_{F}\\
&\leq\|\hat{X}^{*}\hat{V}'(\hat{V}\hat{V}')^{-1}-\hat{X}^{*}V'(VV')^{-1}\|_{F}+\|\hat{X}^{*}V'(VV')^{-1}-X^{*}V'(VV')^{-1}\|_{F}\\
&\leq\|\hat{X}^{*}\|_{F}\|\hat{V}'(\hat{V}\hat{V}')^{-1}-V'(VV')^{-1}\|_{F}+\|\hat{X}^{*}-X^{*}\|_{F}\|V'(VV')^{-1}\|_{F}\\
&\leq\|\hat{X}^{*}\|_{F}\|\hat{V}'(\hat{V}\hat{V}')^{-1}-V'(VV')^{-1}\|_{F}+\|\hat{X}^{*}-X^{*}\|_{F}\|V\|_{F}\|(VV')^{-1}\|_{F}\\
&=\sqrt{n}\|\hat{V}'(\hat{V}\hat{V}')^{-1}-V'(VV')^{-1}\|_{F}+\sqrt{K}\|\hat{X}^{*}-X^{*}\|_{F}\|(VV')^{-1}\|_{F}\\
&\leq \sqrt{n}\|\hat{V}'(\hat{V}\hat{V}')^{-1}-V'(VV')^{-1}\|_{F}+\frac{\sqrt{K}}{m}(\sqrt{K+1}+\frac{K\delta_{\mathrm{max}}}{\lambda_{K}(\tau+\delta_{\mathrm{max}})})err_{n}\|(VV')^{-1}\|_{F},
\end{align*}
to bound $\|\hat{Y}-Y\|_{F}$, we need to bound $\|\hat{V}'(\hat{V}\hat{V}')^{-1}-V'(VV')^{-1}\|_{F}$ and $\|(VV')^{-1}\|_{F}$. First, we bound $\|(VV')^{-1}\|_{F}$ based on below lemma:
\begin{lem}\label{boundinvH}
For any positive definite symmetric matrix $H$, with rank $r$, we have
\begin{align*}
\|H^{-1}\|_{F}\leq \frac{\sqrt{r}}{\lambda_{\mathrm{min}}(H)},
\end{align*}
where $\lambda_{\mathrm{min}}(H)$ is the smallest eigenvalue of $H$.
\end{lem}
\begin{proof}
Since $\|B\|_{F}\leq \sqrt{r}\|B\|$ for any $r\times r$ matrix $B$, we have $\|H^{-1}\|_{F}\leq \sqrt{r}\|H^{-1}\|=\frac{\sqrt{r}}{\lambda_{\mathrm{min}}(H)}$.
\end{proof}
Since $VV'$ is a positive definite symmetric matrix, we have $\|(VV')^{-1}\|_{F}\leq \frac{\sqrt{K}}{\lambda_{\mathrm{min}}(VV')}$ by Lemma \ref{boundinvH}.

Next we bound $\|\hat{V}'(\hat{V}\hat{V}')^{-1}-V'(VV')^{-1}\|_{F}$,
\begin{align*}
&\|\hat{V}'(\hat{V}\hat{V}')^{-1}-V'(VV')^{-1}\|_{F}\\
&=\|\hat{V}'(\hat{V}\hat{V}')^{-1}-\hat{V}'(VV')^{-1}+\hat{V}'(VV')^{-1}-V'(VV')^{-1}\|_{F}\\
&\leq\|\hat{V}'(\hat{V}\hat{V}')^{-1}-\hat{V}'(VV')^{-1}\|_{F}+\|\hat{V}'(VV')^{-1}-V'(VV')^{-1}\|_{F}\\
&\leq\|\hat{V}\|_{F}\|(\hat{V}\hat{V}')^{-1}-(VV')^{-1}\|_{F}+\|\hat{V}-V\|_{F}\|(VV')^{-1}\|_{F}\\
&=\sqrt{K}\|(\hat{V}\hat{V}')^{-1}-(VV')^{-1}\|_{F}+\|\hat{V}-V\|_{F}\|(VV')^{-1}\|_{F}\\
&\leq\sqrt{K}\|(\hat{V}\hat{V}')^{-1}-(VV')^{-1}\|_{F}+\frac{\sqrt{K}\|\hat{V}-V\|_{F}}{\lambda_{\mathrm{min}}(VV')}\\
&=\sqrt{K}\|(\hat{V}\hat{V}')^{-1}(\hat{V}\hat{V}'-VV')(VV')^{-1}\|_{F}+\frac{\sqrt{K}\|\hat{V}-V\|_{F}}{\lambda_{\mathrm{min}}(VV')}\\
&\leq\sqrt{K}\|(\hat{V}\hat{V}')^{-1}\|_{F}\|\hat{V}\hat{V}'-VV'\|_{F}\|(VV')^{-1}\|_{F}+\frac{\sqrt{K}\|\hat{V}-V\|_{F}}{\lambda_{\mathrm{min}}(VV')}\\
&\leq\frac{K^{3/2}\|\hat{V}\hat{V}'-VV'\|_{F}}{\lambda_{\mathrm{min}}(\hat{V}\hat{V}')\lambda_{\mathrm{min}}(VV')}+\frac{\sqrt{K}\|\hat{V}-V\|_{F}}{\lambda_{\mathrm{min}}(VV')}\\
&=\frac{K^{3/2}\|\hat{V}\hat{V}'-\hat{V}V'+\hat{V}V'-VV'\|_{F}}{\lambda_{\mathrm{min}}(\hat{V}\hat{V}')\lambda_{\mathrm{min}}(VV')}+\frac{\sqrt{K}\|\hat{V}-V\|_{F}}{\lambda_{\mathrm{min}}(VV')}\\
&\leq\frac{K^{3/2}\|\hat{V}\hat{V}'-\hat{V}V'\|_{F}+K^{3/2}\|\hat{V}V'-VV'\|_{F}}{\lambda_{\mathrm{min}}(\hat{V}\hat{V}')\lambda_{\mathrm{min}}(VV')}+\frac{\sqrt{K}\|\hat{V}-V\|_{F}}{\lambda_{\mathrm{min}}(VV')}\\
&\leq\frac{K^{3/2}\|\hat{V}\|_{F}\|\hat{V}-V\|_{F}+K^{3/2}\|\hat{V}-V\|_{F}\|V\|_{F}}{\lambda_{\mathrm{min}}(\hat{V}\hat{V}')\lambda_{\mathrm{min}}(VV')}+\frac{\sqrt{K}\|\hat{V}-V\|_{F}}{\lambda_{\mathrm{min}}(VV')}\\
&\leq\frac{2K^{2}\|\hat{V}-V\|_{F}}{\lambda_{\mathrm{min}}(\hat{V}\hat{V}')\lambda_{\mathrm{min}}(VV')}+\frac{\sqrt{K}\|\hat{V}-V\|_{F}}{\lambda_{\mathrm{min}}(VV')}\\
&=(\frac{2K^{2}}{\lambda_{\mathrm{min}}(\hat{V}\hat{V}')}+\sqrt{K})\frac{\|\hat{V}-V\|_{F}}{\lambda_{\mathrm{min}}(VV')}.
\end{align*}
Next we aim to bound $\lambda_{\mathrm{min}}(\hat{V}\hat{V}')$ by $\lambda_{\mathrm{min}}(VV')$. By Weyl's inequality \cite{weyl1912das}, we have
\begin{align*}
|\lambda_{\mathrm{min}}(\hat{V}\hat{V}')-\lambda_{\mathrm{min}}(VV')|\leq\|\hat{V}\hat{V}'-VV'\|.
\end{align*}
Since spectral norm is always smaller than or equal to the Frobenius norm for any matrix, we have $|\lambda_{\mathrm{min}}(\hat{V}\hat{V}')-\lambda_{\mathrm{min}}(VV')|\leq\|\hat{V}\hat{V}'-VV'\|_{F}$. Since
\begin{align*}
\|\hat{V}\hat{V}'-VV'\|_{F}&=\|\hat{V}\hat{V}'-VV'\|_{F}\\
&=\|\hat{V}\hat{V}'-\hat{V}V'+\hat{V}V'-VV'\|_{F}\\
&\leq\|\hat{V}\hat{V}'-\hat{V}V'\|_{F}+\|\hat{V}V'-VV'\|_{F}\\
&\leq\|\hat{V}\|_{F}\|\hat{V}'-V'\|_{F}+\|\hat{V}-V\|_{F}\|V\|_{F}\\
&=2\sqrt{K}\|\hat{V}-V\|_{F},
\end{align*}
we have
\begin{align*}
\lambda_{\mathrm{min}}(VV')-2\sqrt{K}\|\hat{V}-V\|_{F}\leq\lambda_{\mathrm{min}(\hat{V}\hat{V}')}\leq \lambda_{\mathrm{min}}(VV')+2\sqrt{K}\|\hat{V}-V\|_{F}
\end{align*}

Therefore, the bound of $\|\hat{Y}-Y\|_{F}$ is
\begin{align*}
\|\hat{Y}-Y\|_{F}&\leq\sqrt{n}\|\hat{V}'(\hat{V}\hat{V}')^{-1}-V'(VV')^{-1}\|_{F}+\frac{\sqrt{K}}{m}(\sqrt{K+1}+\frac{K\delta_{\mathrm{max}}}{\lambda_{K}(\tau+\delta_{\mathrm{max}})})err_{n}\|(VV')^{-1}\|_{F}\\
&\leq\sqrt{n}(\frac{2K^{2}}{\lambda_{\mathrm{min}}(\hat{V}\hat{V}')}+\sqrt{K})\frac{\|\hat{V}-V\|_{F}}{\lambda_{\mathrm{min}}(VV')}+\frac{K}{m\lambda_{\mathrm{min}}(VV')}(\sqrt{K+1}+\frac{K\delta_{\mathrm{max}}}{\lambda_{K}(\tau+\delta_{\mathrm{max}})})err_{n}\\
&\leq\sqrt{n}(\frac{2K^{2}}{\lambda_{\mathrm{min}}(VV')-2\sqrt{K}\|\hat{V}-V\|_{F}}+\sqrt{K})\frac{\|\hat{V}-V\|_{F}}{\lambda_{\mathrm{min}}(VV')}\\
&~~~+\frac{K}{m\lambda_{\mathrm{min}}(VV')}(\sqrt{K+1}+\frac{K\delta_{\mathrm{max}}}{\lambda_{K}(\tau+\delta_{\mathrm{max}})})err_{n}\\
&\leq\sqrt{n}(\frac{2K^{2}}{m_{V}-2\sqrt{K}\|\hat{V}-V\|_{F}}+\sqrt{K})\frac{\|\hat{V}-V\|_{F}}{m_{V}}\\
&~~~+\frac{K}{mm_{V}}(\sqrt{K+1}+\frac{K\delta_{\mathrm{max}}}{\lambda_{K}(\tau+\delta_{\mathrm{max}})})err_{n}\\
&\leq(\frac{2K^{2}}{m_{V}-2\sqrt{K}\frac{CK}{m\sqrt{n}}(\sqrt{K+1}+\frac{K\delta_{\mathrm{max}}}{\lambda_{K}(\tau+\delta_{\mathrm{max}})})err_{n}}+\sqrt{K})\frac{CK(\sqrt{K+1}+\frac{K\delta_{\mathrm{max}}}{\lambda_{K}(\tau+\delta_{\mathrm{max}})})err_{n}}{mm_{V}}\\
&~~~+\frac{K}{mm_{V}}(\sqrt{K+1}+\frac{K\delta_{\mathrm{max}}}{\lambda_{K}(\tau+\delta_{\mathrm{max}})})err_{n}\\
&= Err_{n}.
\end{align*}
\end{proof}
\subsection{Proof of Theorem \ref{main}.}
\begin{proof}
The difference between the row-normalized projection coefficients $\Pi$ and $\hat{\Pi}$ can be bounded by the difference between $Y$ and $\hat{Y}$, since
\begin{align*}
\|\hat{\pi}_{i}-\pi_{i}\|_{F}&=\|\frac{\hat{Y}_{i}}{\|\hat{Y}_{i}\|_{F}}-\frac{Y_{i}}{\|Y_{i}\|_{F}}\|_{F}\\
&=\|\frac{\hat{Y}_{i}\|Y_{i}\|_{F}-Y_{i}\|\hat{Y}_{i}\|_{F}}{\|\hat{Y}_{i}\|_{F}\|Y_{i}\|_{F}}\|_{F}\\
&=\|\frac{\hat{Y}_{i}\|Y_{i}\|_{F}-\hat{Y}_{i}\|\hat{Y}_{i}\|_{F}+\hat{Y}_{i}\|\hat{Y}_{i}\|_{F}-Y_{i}\|\hat{Y}_{i}\|_{F}}{\|\hat{Y}_{i}\|_{F}\|Y_{i}\|_{F}}\|_{F}\\
&\leq \frac{\|\hat{Y}_{i}\|Y_{i}\|_{F}-\hat{Y}_{i}\|\hat{Y}_{i}\|_{F}\|_{F}+\|\hat{Y}_{i}\|\hat{Y}_{i}\|_{F}-Y_{i}\|\hat{Y}_{i}\|_{F}\|_{F}}{\|\hat{Y}_{i}\|_{F}\|Y_{i}\|_{F}}\\
&=\frac{\|\hat{Y}_{i}\|_{F}|\|Y_{i}\|_{F}-\|\hat{Y}_{i}\|_{F}|+\|\hat{Y}_{i}\|_{F}\|\hat{Y}_{i}-Y_{i}\|_{F}}{\|\hat{Y}_{i}\|_{F}\|Y_{i}\|_{F}}\\
&=\frac{|\|Y_{i}\|_{F}-\|\hat{Y}_{i}\|_{F}|+\|\hat{Y}_{i}-Y_{i}\|_{F}}{\|Y_{i}\|_{F}}\\
&\leq\frac{2\|\hat{Y}_{i}-Y_{i}\|_{F}}{\|Y_{i}\|_{F}}\\
&\leq \frac{2\|\hat{Y}_{i}-Y_{i}\|_{F}}{\mathrm{min}_{i}\|Y_{i}\|_{F}},
\end{align*}
we have
\begin{align*}
\frac{\|\hat{\Pi}-\Pi\|_{F}}{\sqrt{n}}\leq \frac{2\|\hat{Y}-Y\|_{F}}{\mathrm{min}_{i}\|Y_{i}\|_{F}\sqrt{n}}\leq \frac{2Err_{n}}{\mathrm{min}_{i}\|Y_{i}\|_{F}\sqrt{n}}=\frac{2Err_{n}}{m_{Y}\sqrt{n}}.
\end{align*}
\end{proof}
\end{document}